\documentclass[11pt,a4paper,twosided]{article}

\usepackage{geometry}
\usepackage[utf8]{inputenc}
\usepackage[T1]{fontenc}

\usepackage[english]{babel}

\usepackage{microtype}

\usepackage{amsmath,amsthm,amssymb,stmaryrd,thmtools,upgreek}

\newtheorem{theorem}{Theorem}
\newtheorem{lemma}[theorem]{Lemma}
\newtheorem{corollary}[theorem]{Corollary}

\newtheorem{conjecture}[theorem]{Conjecture}

\theoremstyle{definition}
\newtheorem{definition}[theorem]{Definition}
\newtheorem{remark}[theorem]{Remark}
\newtheorem{example}[theorem]{Example}

\usepackage{graphicx}
\usepackage[obeyclassoptions,mode=tex]{standalone}
\usepackage{tikz}
\usetikzlibrary{backgrounds}
\usetikzlibrary{shapes.geometric}
\usetikzlibrary{positioning}
\usetikzlibrary{automata}
\usetikzlibrary{tikzmark}
\usetikzlibrary{patterns}
\usetikzlibrary{arrows}
\tikzset{every state/.style={minimum size=1pt}}
\usepackage{tikz-cd}

\usepackage{pgfornament}

\usepackage{hyperref}
\usepackage[capitalise,noabbrev,nameinlink]{cleveref}
\Crefname{assumption}{Assumption}{Assumptions}
\usepackage[electronic,hyperref,xcolor,cleveref]{knowledge}
\knowledgeconfigure{notion}

\usepackage{booktabs}
\usepackage{varwidth}

\usepackage{algorithm2e}
\Crefname{algocfline}{Algorithm}{Algorithms}
\crefname{algocfline}{Algorithm}{Algorithms}
\crefname{algocf}{Algorithm}{Algorithms}
\Crefname{algocf}{Algorithm}{Algorithms}
\crefname{conjecture}{Conjecture}{Conjectures}
\Crefname{conjecture}{Conjecture}{Conjectures}

\usepackage{xparse}
\usepackage{xpatch}
\usepackage{tokcycle}
\usepackage{ifthen}

\usepackage{bussproofs}

\usepackage{ensps-colorscheme}
\knowledgestyle{intro notion}{color={A2}, emphasize}
\knowledgestyle{notion}{color={B1}}
\hypersetup{
    colorlinks=true,
    breaklinks=true,
    linkcolor=A4,
    citecolor=A4,
    urlcolor=A4,
    filecolor=A4,
}

\newcommand{\klogo}{%
\begin{tikzpicture}[scale=0.2,line/.style={draw, line width=0.2pt, line cap=round, line join=round}]
\coordinate (A00) at (0,0);
\coordinate (A01) at (0,1);
\coordinate (A10) at (1,0);
\coordinate (B10) at (1,0.2);
\coordinate (B01) at (0.2,1);

\coordinate (C01) at (0.4,0.7);
\coordinate (C10) at (0.7,0.4);
\coordinate (C12) at (0.4,1.2);
\coordinate (C21) at (1.2, 0.4);
\coordinate (C22) at (1.2, 1.2);

\coordinate (D00) at (C10);
\coordinate (D01) at (0.8,0.5);
\coordinate (D10) at (0.8,0.3);

\coordinate (E01) at (0.3,0.7);
\coordinate (E10) at (0.5,0.7);

\draw[line] (B01) -- (A01) -- (A00) -- (A10) -- (B10);
\draw[line] (C01) -- (C12) -- (C22) -- (C21) -- (C10);

\draw[line] (D01) -- (D00) -- (D10);
\draw[line] (E01) -- (E10);

\end{tikzpicture}%
}

\makeatletter
\newcommand\mathgr[1]{\tokcycle
  {\addcytoks{##1}}
  {\processtoks{##1}}
  {\ifcsname up\expandafter\@gobble\string##1\endcsname
   \addcytoks[1]{\csname up\expandafter\@gobble\string##1\endcsname}%
    \else\addcytoks{##1}\fi}
  {\addcytoks{##1}}{#1}%
  \expandafter\mathrm\expandafter{\the\cytoks}%
}
\makeatother

\NewDocumentEnvironment{proofof}{ m O{appendix} }{
    \ifcsname #1\endcsname
        \def\isInsideRestatedTheorem{1}
        \csname #1\endcsname*
    \fi
    \begin{proof}[Proof of {\cref{#1}} as stated on page {\pageref{#1}}]
        \phantomsection
        \label{#1:proof}
}{
        \ifthenelse{\equal{#2}{appendix}}{
        \marginpar{\vspace{-2em}\texttt{\small{\hyperref[#1]{$\triangleright$ Back to p.\pageref{#1}}}}}
        }{}
    \end{proof}
}

\NewDocumentCommand{\proofref}{ m }{
    \IfRefUndefinedExpandable{#1:proof}{}{
        \ifdefined\isInsideRestatedTheorem
        \else
            \marginpar{\vspace{0.6em}\texttt{\small{\hyperref[#1:proof]{$\triangleright$ Proven p.\pageref{#1:proof}}}}}
        \fi
    }
}

\newcommand{\Grbs}{\kl{Gröbner bases}}
\newcommand{\Grb}{\kl{Gröbner basis}}

\NewDocumentCommand{\set}{ m }{\{ #1 \}}
\NewDocumentCommand{\setof}{ m m }{\{ #1 \mid #2 \}}

\NewDocumentCommand{\seqof}{ m O{n \in \N} }{\left( #1 \right)_{#2}}

\NewDocumentCommand{\defined}{ }{\triangleq}

\newcommand{\subfin}{\subset_{\text{fin}}}
\newcommand{\subseteqfin}{\subseteq_{\text{fin}}}

\NewDocumentCommand{\EXPTIME}{}{\ensuremath{\mathsf{EXPTIME}}}

\NewDocumentCommand{\range}{ O{1} m }{[#1, #2]}

\newcommand{\tobij}{\stackrel{\simeq}{\longrightarrow}}

\NewDocumentCommand{\NewDocumentOrdering}{ m m m }{
    \expandafter\newcommand\csname #1leq\endcsname{
        \mathrel{\kl[#1]{#2}}
    }
    \expandafter\newcommand\csname #1lt\endcsname{
        \mathrel{\kl[#1]{#3}}
    }
    \knowledge{#1}{notion}
}

\NewDocumentCommand{\upset}{ O{} m }{{\uparrow_{#1} #2}}
\NewDocumentCommand{\dwset}{ O{} m }{{\downarrow_{#1} #2}}

\NewDocumentCommand{\factorial}{ O{} m }{
    \if\relax\detokenize{#1}\relax
        #2!
    \else
        (#2)!
    \fi
}

\newcommand{\D}{\mathcal{D}}
\newcommand{\G}{\mathcal{G}}

\newcommand{\Q}{\mathbb{Q}}
\newcommand{\Z}{\mathbb{Z}}
\newcommand{\N}{\mathbb{N}}
\newcommand{\K}{\mathbb{K}}
\newcommand{\X}{\mathcal{X}}
\newcommand{\Y}{\mathcal{Y}}

\newcommand{\T}{\mathcal{T}}

\newcommand{\ftwo}{\mathbb{F}_2}

\newcommand{\calZ}{\mathcal{Z}}
\newcommand{\calH}{\mathcal{H}}

\newcommand{\calM}{\mathcal{M}}

\newcommand{\poly}[2]{#1[#2]}
\newcommand{\mon}[2][]{\mathop{\kl[\mon]{\mathsf{Mon}_{#1}}}(#2)}
\knowledge{\mon}{notion}
\NewDocumentCommand{\monelt}{ O{m} }{\mathfrak{#1}}

\NewDocumentCommand{\divleq}{}{
  \mathrel{\kl[\divleq]{\sqsubseteq^{\mathrm{div}}}}
}
\knowledge{\divleq}{notion}

\NewDocumentCommand{\gdivleq}{ O{\group} }{
  \mathrel{\kl[\gdivleq]{\sqsubseteq^{\mathrm{div}}_{#1}}}
}
\knowledge{\gdivleq}{notion}

\NewDocumentOrdering{var}{\leq_\X}{<_\X}

\NewDocumentOrdering{revlex}{\sqsubseteq^{\mathsf{RevLex}}}{\sqsubset^{\mathsf{RevLex}}}
\NewDocumentOrdering{grevlex}{\sqsubseteq^{\mathsf{RevLex}}_\group}{\sqsubset^{\mathsf{RevLex}}_\group}
\NewDocumentOrdering{lex}{\sqsubseteq_{\mathsf{Lex}}}{\sqsubset_{\mathsf{Lex}}}
\NewDocumentOrdering{shortlex}{\sqsubseteq_{\mathsf{ShortLex}}}{\sqsubset_{\mathsf{ShortLex}}}

\NewDocumentOrdering{pmon}{\preceq}{\prec}

\newcommand{\lc}[1][]{\mathop{\kl[\lc]{\mathsf{LC}_{#1}}}}
\knowledge{\lc}{notion}
\newcommand{\lt}[1][]{\mathop{\kl[\lt]{\mathsf{LT}_{#1}}}}
\knowledge{\lt}{notion}
\newcommand{\lm}[1][]{\mathop{\kl[\lm]{\mathsf{LM}_{#1}}}}
\knowledge{\lm}{notion}
\newcommand{\cm}[1][]{\mathop{\kl[\cm]{\mathsf{CM}_{#1}}}}
\knowledge{\cm}{notion}

\newcommand{\dom}{\mathop{\kl[\dom]{\mathsf{dom}}}}
\knowledge{\dom}{notion}

\newcommand{\lmdec}[1][]{\mathop{\kl[\lmdec]{\mathsf{LM}_{#1}}}}
\knowledge{\lmdec}{notion}
\newcommand{\domdec}{\mathop{\kl[\domdec]{\mathsf{dom}}}}
\knowledge{\domdec}{notion}

\newcommand{\lcm}{\mathop{\kl[\lcm]{\mathsf{LCM}}}}
\knowledge{\lcm}{notion}

\NewDocumentCommand{\idl}{O{I}}{\mathcal{#1}}
\NewDocumentCommand{\IdlGen}{ m }{\withkl{\kl[\IdlGen]}{
  \mathopen{\cmdkl{\langle}}
  #1
\mathclose{\cmdkl{\rangle}}}}
\knowledge{\IdlGen}{notion}

\NewDocumentCommand{\EqIdlGen}{ O{\group} m }{\withkl{\kl[\EqIdlGen]}{
  \mathopen{\cmdkl{\langle}}
  #2 
  \mathclose{\cmdkl{\rangle}}_{#1}}}
\knowledge{\EqIdlGen}{notion}

\newcommand{\group}{\mathcal{G}}
\newcommand{\grp}[1][G]{\mathcal{#1}}
\newcommand{\actson}{\curvearrowright}
\NewDocumentCommand{\gelem}{ O{\pi} }{\mathgr{#1}}

\newcommand{\orbit}[2][]{\mathop{\kl[\orbit]{\mathsf{orbit}_{#1}}}(#2)}
\knowledge{\orbit}{notion}

\newcommand{\inc}[1]{\textsc{Inc}_{#1}}

\NewDocumentCommand{\Basis}{O{B}}{\mathcal{#1}}
\NewDocumentCommand{\LBasis}{O{B} m}{\mathcal{#1}_{#2}}

\NewDocumentCommand{\Indets}{}{\mathcal{X}}
\NewDocumentCommand{\IndetsCol}{}{\mathcal{Y}}
\NewDocumentCommand{\IndetsV}{ O{V} }{\mathcal{X}_{#1}}
\NewDocumentCommand{\idlZ}{}{\mathcal{Z}}

\newcommand{\ordfin}[1]{\kl[\ordfin]{\mathgr{#1}}}
\newcommand{\om}{\kl[\om]{\omega}}
\newcommand{\ordplus}{\mathbin{\kl[\ordplus]{+}}}
\newcommand{\ordtimes}{\mathbin{\kl[\ordtimes]{\times}}}
\knowledge{\ordplus}{notion}
\knowledge{\ordtimes}{notion}
\knowledge{\ordfin}{notion}
\knowledge{\om}{notion}

\NewDocumentCommand{\toeucl}{m}{
  \mathrel{\kl[\toeucl]{\to_{#1}}}
}
\knowledge{\toeucl}{notion}

\NewDocumentCommand{\CancelPoly}{ m m }{\kl[\CancelPoly]{\mathsf{C}_{#1,#2}}}
\knowledge{\CancelPoly}{notion}

\newcommand{\spoly}[2]{\mathop{\kl[\spoly]{\mathsf{S}}}(#1,#2)}
\knowledge{\spoly}{notion}
\newcommand{\spolyset}{\mathop{\kl[\spolyset]{\mathsf{SSet}}}}
\knowledge{\spolyset}{notion}
\newcommand{\rem}[3][]{\kl[\rem]{\mathsf{Rem}^{#1}_{#2}(#3)}}
\knowledge{\rem}{notion}

\NewDocumentCommand{\forgetCol}{}{\mathop{\kl[\forgetCol]{\mathsf{forget}}}}
\knowledge{\forgetCol}{notion}
\NewDocumentCommand{\colorWith}{ m }{\mathop{\kl[\colorWith]{\mathsf{col}_{#1}}}}
\knowledge{\colorWith}{notion}
\NewDocumentCommand{\freeColor}{ }{\mathop{\kl[\freeColor]{\mathsf{freecol}}}}
\knowledge{\freeColor}{notion}

\newcommand{\wenc}[1]{\withkl{\kl[\wenc]}{\mathopen{\cmdkl{\llbracket}} #1 \mathclose{\cmdkl{\rrbracket}}}}
\knowledge{\wenc}{notion}

\newcommand{\lin}{\textsc{Lin}}

\newcommand{\aut}[2][]{\mathsf{Aut}_{#1}{(#2)}}

\newcommand{\EqualityAtoms}{\kl[\EqualityAtoms]{\mathcal{A}}}
\knowledge{\EqualityAtoms}{notion}

\newcommand{\OrderAtoms}{\kl[\OrderAtoms]{\mathcal{Q}}}
\knowledge{\OrderAtoms}{notion}

\newcommand{\RadoAtoms}{\kl[\RadoAtoms]{\mathcal{G}}}
\knowledge{\RadoAtoms}{notion}

\newcommand{\BitVectorAtoms}{\kl[\BitVectorAtoms]{\mathcal{V}}}
\knowledge{\BitVectorAtoms}{notion}

\newcommand{\TreeAtoms}{\kl[\TreeAtoms]{\mathcal{T}}}
\knowledge{\TreeAtoms}{notion}

\newcommand{\lexGroupAction}{\mathbin{\kl[\lexGroupAction]{\otimes}}}
\knowledge{\lexGroupAction}{notion}

\knowledge{notion}
 | group action

\knowledge{notion}
 | group
\knowledge{notion}
 | orbits
 | orbit
 | $\group _V$-orbits

\knowledge{notion}
 | oligomorphic
\knowledge{notion}
 | effectively oligomorphic
 | effective oligomorphicity
 | effective oligomorphism
 | effectively oligomorphically

\knowledge{notion}
 | orbit-finite
 | orbit finite
 | orbit finite set
 | orbit finite sets

\knowledge{notion}
 | respects a linear order
 | respect a linear order
 | respected by the group action
 | respects a computable linear order
 | comatible with a linear order
 | compatible@ord
 | compatible with the order@ord

 \knowledge{notion}
 | $V$-compatible 

\knowledge{notion}
 | polynomials
 | polynomial

\knowledge{notion}
 | ordinal
 | ordinals

\knowledge{notion}
 | divides
 | divisibility order
 | divisibility

\knowledge{notion}
 | leading monomial
 | leading monomials
\knowledge{notion}
 | leading coefficient
\knowledge{notion}
 | leading term
\knowledge{notion}
 | characteristic monomial
 | characteristic monomials
\knowledge{notion}
 | domain@polynomial
 | domain@poly
 | domain@monomial
 | domain of a polynomial

\knowledge{notion}
 | decomposition
 | decompositions

\knowledge{notion}
 | leading monomial@dec
 | leading monomial of the decomposition
 | leading monomial of a decomposition
\knowledge{notion}
 | domain of the decomposition
 | domain of a decomposition
 | domain@dec

\knowledge{notion}
 | divisibility up-to $\group $
 | divisibility relation up-to-$\group $
 | divisibility up to $\grp [G] \times \grp [H]$

\knowledge{notion}
 | $T$-strong basis
 | strong basis

\knowledge{notion}
 | reverse lexicographic 
 | reverse lexicographic order
 | reverse lexicographic ordering
 | revlex ordering
 | revlex order

\knowledge{notion}
 | well-quasi-ordered
 | well-quasi-orders
 | well-quasi-ordering
 | well-quasi-order
 | well-quasi-orderings
 | WQO
 | WQOs

\knowledge{notion}
 | effective total ordering

\knowledge{notion}
 | computability assumptions

\knowledge{notion}
 | Noetherian
 | Noetherian ring
 | Noetherian rings

\knowledge{notion}
 | ideal membership problem
\knowledge{notion}
 | equivariant ideal membership problem
 | equivariant ideal membership
 | equivariant ideal membership problems
 | equivariant ideal membership queries

\knowledge{notion}
 | ideal
 | ideals

\knowledge{notion}
 | ideal generated by
 | generated by@idl
 | generates@idl

\knowledge{notion}
 | equivariance
 | equivariant
 | equivariant set
 | invariant
 | equivariant subset
 | equivariant subsets

\knowledge{notion}
 | equivariant ideal
 | equivariant ideals
 | equivariant@ideal

\knowledge{notion}
 | equivariant function
 | equivariant@func
 | equivariant operation

\knowledge{notion}
 | generated by@eqidl
 | equivariant ideal generated by
 | orbit-finitely generated

\knowledge{notion}
 | Hilbert's Basis Theorem
 | Hilbert's basis theorem
 | Hilbert Basis Theorem
 | Hilbert Basis Property
 | Hilbert's basis property
 | Hilbert basis property

\knowledge{notion}
 | Equivariant Hilbert basis property
 | equivariant Hilbert basis property

\knowledge{notion}
 | Gröbner basis
 | Gröbner bases
 | Gr\"{o}bner basis
 | Gr\"{o}bner bases

\knowledge{notion}
 | equivariant Gröbner basis
 | equivariant Gröbner bases
 | equivariant Gr\"{o}bner basis
 | equivariant Gr\"{o}bner bases

\knowledge{notion}
 | weak equivariant Gröbner basis
 | weak equivariant Gröbner bases

\knowledge{notion}
 | $V$-strong equivariant Gröbner basis
 | $V$-strong equivariant Gröbner bases

\knowledge{notion}
 | Buchberger algorithm
 | Buchberger's algorithm

\knowledge{notion}
 | $S$-polynomials
 | $S$-polynomial
 | S-polynomials
 | S-polynomial

\knowledge{notion}
 | normalised
\knowledge{notion}
 | terminating 
\knowledge{notion}
 | remainders
 | remainder

\knowledge{notion}
 | monomial reachability problem
\knowledge{notion}
 | monomial rewrite system
 | monomial rewriting system
 | monomial rewriting
 | monomial rewrite systems
\knowledge{notion}
 | rewrite step@monrew

\knowledge{notion}
 | word encodings
 | word encoding

\knowledge{notion}
 | define paths@efo
 | defines a path@efo

\knowledge{notion}
 | infinite path@of
 | infinite paths@of
 | finite path@of
 | finite paths@of
\knowledge{notion}
 | infinite dimensional vector space

\knowledge{notion}
 | Rado graph

\knowledge{notion,typewriter,ensuretext}
 | weakgb
\knowledge{notion,typewriter,ensuretext}
 | egb
\knowledge{notion}
 | topology
 | topological space
 | <F5><F5>topological spaces
\knowledge{notion}
 | open subsets
 | open subset
\knowledge{notion}
 | closed subsets
 | closed subset

\knowledge{notion}
 | non-homogeneous
 | homogeneous
 | homogeneous structures
 | homogeneous structure
\knowledge{notion}
 | nicely orderable

\knowledge{notion}
 | reduct
\knowledge{notion}
 | effective reduct
\knowledge{notion}
 | sum action
\knowledge{notion}
 | product action
\knowledge{notion}
 | lexicographic product action

\knowledge{notion}
 | Noetherian@space
 | Noetherian spaces
 | Noetherian space

\knowledge{notion}
 | well-structured transition systems
\knowledge{notion}
 | topological well-structured transition system
 | topological well-structured transition systems

\knowledge{notion}
 | open reachability problem
\knowledge{notion}
 | backward algorithm

\knowledge{notion}
 | Zariski topology
\knowledge{notion}
 | equivariant Zariski topology

\knowledge{notion}
 | polynomial automata
 | polynomial automaton
\knowledge{notion}
 | zeroness problem@pa
 | zeroness problem for polynomial automata
 | zeroness@pa
\knowledge{notion}
 | orbit finite polynomial automata
 | orbit finite polynomial automaton
 | Orbit finite polynomial automata
\knowledge{notion}
 | zeroness problem@ofpa

\title{Computability of Equivariant Gröbner bases}

\author{
Arka Ghosh\thanks{Université de Bordeaux}\thanks{University of Warsaw}
 \and
Aliaume Lopez\thanks{University of Warsaw}
}

\newcommand{\repositoryUrl}{\url{https://github.com/AliaumeL/AtomicHilbert}}

\knowledge{notion}
 | kl-usage

\begin{document}
\maketitle
\begin{abstract}
    Let \(\mathbb{K}\) be a field, \(\mathcal{X}\) be an infinite set (of indeterminates), and \(\mathcal{G}\) be a group acting on \(\mathcal{X}\). An ideal in the polynomial ring \(\mathbb{K}[\mathcal{X}]\) is called equivariant if it is invariant under the action of \(\mathcal{G}\). We show Gröbner bases for equivariant ideals are computable are hence the equivariant ideal membership is decidable when \(\mathcal{G}\) and \(\mathcal{X}\) satisfies the Hilbert's basis property, that is, when every equivariant ideal in \(\mathbb{K}[\mathcal{X}]\) is finitely generated. Moreover, we give a sufficient condition for the undecidability of the equivariant ideal membership problem. This condition is satisfied by the most common examples not satisfying the Hilbert's basis property.
\paragraph{Repository:} \repositoryUrl
\end{abstract}

\klogo\ This document uses \href{https://ctan.org/pkg/knowledge}{knowledge}:
\kl[kl-usage]{notion} points to its \intro[kl-usage]{definition}.

\section{Introduction}
\label{sec:intro}

\AP For a field $\K$ and a non-empty set $\Indets$ of indeterminates, we use
$\poly{\K}{\Indets}$ to denote the ring of polynomials with coefficients from $\K$
and indeterminates/variables from $\Indets$. A fundamental result in commutative
algebra is \intro{Hilbert's basis theorem}, stating that when $\Indets$ is finite,
every ideal in $\poly{\K}{\Indets}$ is finitely generated \cite{HILB1890}, where an
\kl{ideal} is a non-empty subset of $\poly{\K}{\Indets}$ that is closed under
addition and multiplication by elements of $\poly{\K}{\Indets}$. This property can
be rephrased as the fact that the set of polynomials $\poly{\K}{\Indets}$ is
\intro{Noetherian}. \kl{Hilbert's basis theorem} extends to the case where $\K$
is a ring that is itself \kl{Noetherian} \cite[Theorem 4.1]{Lang02}.

\AP A \Grb\ is a specific kind of generating set of a polynomial ideal
which allows easy checking of membership of a given polynomial in that ideal.
\kl{Gr\"{o}bner bases} were introduced by Buchberger who showed when $\Indets$ is
finite, every ideal in $\poly{\K}{\Indets}$ has a finite \kl{Gr\"{o}bner basis} and
that, for a given a set of polynomials in $\poly{\K}{\Indets}$, one can compute a
finite \kl{Gröbner basis} of the ideal generated by them via the so-called
\intro{Buchberger algorithm} \cite{BUCH76}. The
existence and computability of \Grbs\ implies the decidability of the
\kl{ideal membership problem}: given a polynomial $f$ and set of polynomial
$H$, decide whether $f$ is in the ideal generated by $H$. The theory of
\kl{Gr\"{o}bner bases} has applications in very diverse areas of computer
science, including integer programming \cite{Sturmfels96}, algebraic proof
systems \cite{algProof}, geometric reasoning \cite{Cox2015chGeom}, fixed
parameter tractability \cite{ACDM22}, program analysis \cite{SSM04} and
constraint satisfaction problems \cite{Mas21}.
In automata theory it has been used for deciding zeroness of polynomial
automata \cite{BEDUSHWO17}, reachability in symmetric Petri nets \cite{MAME82},
equivalence for string-to-string transducers \cite{HONKALA00} and equivalence
of polynomial differential equations \cite{CLEMENTE24}. 

\AP There has been a growing interest in the last few years for computational
models that are manipulating infinite data structures in a finite way, for
instance an automaton reading words on the infinite alphabet $\N$, while
maintaining a finite number of states. While this idea can be traced back to
the 90s with the notion of register automata \cite{KAFR94}, it has been revived
in with the development of the theory of \emph{orbit finite sets}. In this
setting, one would like to consider an infinite set of variables $\Indets$. As an
example, let us consider the set $\Indets$ of variables $x_i$ for $i \in \N$, and
the \kl{ideal} $\idlZ$ generated by the set $\setof{x_i}{i \in \N}$. It is
clear that $\idlZ$ is not finitely generated, and we conclude that the
\kl{Hilbert's basis theorem} (and a fortiori, the \kl{Gr\"{o}bner basis}
theory) does not extend to the case of infinite sets of indeterminates.

\AP However, in the applications mentionned above, the infinite set of
variables (data) comes with an extra structure: the behaviour of the considered
systems are invariant under the action of a group $\group$ on $\Indets$. The action
of this $\group$ on $\Indets$ naturally induces an action on $\poly{\K}{\Indets}$, by
renaming the variables. The typical example is the group of all permutations of
$\Indets$, which corresponds to seeing $\Indets$ as a set of \emph{indistinguishable}
names: one is not interested in the ideal $\idlZ$ generated by the set
$\setof{x_i}{i \in \N}$, but rather in the \kl{equivariant ideal} generated by
the set $\setof{x_i}{i \in \N}$, which is the smallest ideal that contains it
and is invariant under the action of $\group$. In this case, this ideal is
finitely generated by a single indeterminate, e.g. $x_1$. Please note that
equivariance does not imply finite generation in general: for instance, the
ideal $\idlZ$ is not finitely generated as an equivariant ideal with respect to
the trivial group.

\subsection{Related Research}
The above-mentioned results were rediscovered in \cite{AH07,AH08,HKL18}. In
\cite{HS12} these results were used to prove the Independent Set Conjecture in
algebraic statistics. In \cite{HS12}, the authors also showed that one can even
take a submonoid $\calM$ of $\inc{<}$ and prove existence and computability of
finite Gr\"{o}bner basis assuming that $\gdivleq[\calM]$ is a
well-partial-order. These results were significantly generalised in
\cite{GHOLAS24}, which gives a necessary and a sufficient condition on the
actions $\group\actson\Indets$ for the \kl{Equivariant Hilbert basis property}
to hold \cite[Theorems 11 and 12, Lemma 13]{GHOLAS24}. The necessary and
sufficient conditions are equivalent up to a well-known conjecture by Pouzet
\cite[Problems 12]{POUZ24}. But to obtain decision procedures, one still lacks
a generalisation of \kl{Buchberger's algorithm} to the equivariant case, except
under artificial extra assumptions \cite[Section 6]{GHOLAS24}. Overall, a
general understanding of the decidability of the \kl{equivariant ideal
membership problem} is still missing, and \emph{a fortiori}, a generalisation
of \kl{Buchberger's algorithm} to the equivariant case is still an open
problem.

Our results are part of a larger research direction that aims at establishing
an algorithmic theory of computation with orbit-finite sets. For instance,
\cite{BFKM24} studies equivariant subspaces of vector spaces generated by
orbit-finite sets, \cite{GHL22,GHL25} study solvability of orbit-finite systems
of linear equations and inequalities, and \cite{BFKM24,GHL22,Prz23} study duals
of vector spaces generated by orbit-finite sets.

\subsection{Contributions.}
\AP In this paper, we bridge the gap between the
theoretical understanding of the \intro{equivariant Hilbert basis property}
 \cite[Property 4]{GHOLAS24}, and the computational aspects of \kl{equivariant
ideals}, by showing that under mild assumptions on the group action, one can
compute an \kl{equivariant Gröbner basis} of an \kl{equivariant ideal}, hence,
that one can decide the \kl{equivariant ideal membership problem}. In order to
compute such sets, we will need to introduce some classical \kl{computability
assumptions} on the group action $\group \actson \Indets$, and on the set of
indeterminates $\Indets$. These will be defined in
\cref{sec:preliminaries}, but informally, we assume
that one can compute representatives of the orbits of elements under the action
of $\group$ (this is called \kl{effective oligomorphism}), and that one has
access to a total ordering on $\Indets$ that is computable, and
\kl(ord){compatible} with the action of $\group$. Please note that the ordering
on $\Indets$ is not required to be well-founded, and a typical example of our
computable assumptions would be the set $\Q$ of rationals, equipped with the
natural ordering $\leq$ and the group $\group$ would be the group of all
monotone bijections from $\Q$ to itself.

\AP Let us now focus on the mild semantic assumption that we will need to make
on the set of indeterminates $\Indets$ and the group $\group$, that will
guarantee the termination of our procedures. We refer to our preliminaries
(\cref{sec:preliminaries}) for a more detailed
discussion on these assumptions, but again informally, we ask that the set of
\kl{monomials} $\mon{\Indets}$ is well-behaved with respect to divisibility up
to the action of $\group$, which we write as the fact that $(\mon{\Indets},
\gdivleq)$ is a \kl{well-quasi-ordering} (\kl{WQO}). It is known from that this
is a necessary condition for the \kl{equivariant Hilbert basis property}
\cref{thm:equiv-hilbert-property}, and we will rely on a slightly stronger
condition, namely that $(\mon[Y]{\Indets}, \gdivleq)$ is a \kl{WQO}, whenever
$(Y, \leq)$ is one, which is conjectured to be equivalent to the first
condition. Beware that \cref{thm:equiv-hilbert-property,thm:compute-egb}
are
incomparable: the former does not talk about decidability, while the latter 
only considers \kl{equivariant ideals} that are already finitely presented, and we 
will show in
\cref{ex:non-wqo-undecidable} an example where \kl{equivariant
Gröbner bases} are computable, but the \kl{Hilbert basis property} fails.

\begin{theorem}[name={\cite[Theorem 11]{GHOLAS24}}]
  \label{thm:equiv-hilbert-property}
  Let $\Indets$ be a totally ordered set of indeterminates
  equipped with a group action $\group \actson \Indets$ that is 
  \kl(ord){compatible} with the ordering on $\Indets$.
  Then, $(\mon[\om]{\Indets}, \gdivleq)$ is a \kl{WQO}, if and only if 
  the \kl{equivariant Hilbert basis property} holds for $\poly{\K}{\Indets}$.
\end{theorem}

\begin{theorem}[name={Equivariant Gröbner Basis},restate=thm:compute-equiv-gb]
  \label{thm:compute-egb}
  Let $\Indets$ be a totally ordered set of indeterminates
  equipped with a group action $\group \actson \Indets$, under our \kl{computability assumptions}.
  If $(\mon[Y]{\Indets}, \gdivleq)$ is a \kl{WQO} for every 
  \kl{well-quasi-ordered} set $(Y,\leq)$, then one can
  compute an \kl{equivariant Gröbner bases} of \kl{equivariant ideals}.
\end{theorem}

\AP To prove our \cref{thm:compute-egb}, we will first introduce a weaker
notion of \kl{weak equivariant Gröbner basis}, which characterises the results
obtained by naïvely adapting \kl{Buchberger's algorithm} to the equivariant
case. Then, we will show that under our \kl{computability assumptions}, one can
start from a finite set of generators $H$ of an \kl{equivariant ideal}, and
compute a well-chosen \kl{weak equivariant Gröbner basis}, which happens to be
an \kl{equivariant Gröbner basis} of the ideal generated by $H$. As a
consequence, we obtain effective representations of \kl{equivariant ideals},
over which one can check membership, inclusion, and compute the sum and
intersection of \kl{equivariant ideals}
(\cref{cor:equivariant-ideals-computations}).

\AP We then focus on providing undecidability results for the \kl{equivariant
ideal membership problem} in the case where our effective assumptions are
satisfied, but the \kl{well-quasi-ordering} condition is not. This aims at
illustrating the fact that our assumptions are close to optimal. One classical
way for a set of structures to not be \kl{well-quasi-ordered} (when labelled
using integers) is to have the ability to represent an \emph{infinite path} (a
formal definition will be given in
\cref{sec:undecidability}). We prove that
whenever one can (effectively) represent an infinite path in the set of
\kl{monomials} $\mon{\Indets}$, then the \kl{equivariant ideal membership
problem} is undecidable.

\begin{theorem}[name={Undecidability of Equivariant Ideal Membership},restate=thm:undecidable-paths]
  \label{thm:undecidable-paths}
  Let $\Indets$ be a totally ordered set of indeterminates
  equipped with a group action $\group \actson \Indets$, under our \kl{computability assumptions}.
  If $\Indets$ contain an \kl(of){infinite path}
  then the \kl{equivariant ideal membership problem} is undecidable.
\end{theorem}

Finally, we illustrate how our positive results find applications in numerous
situations. This is done by providing families indeterminates that satisfy our
\kl{computability assumptions}, and for which we can compute \kl{equivariant
Gröbner bases}, and also by showing how our results can be used in the context
of \kl{topological well-structured transition systems} \cite{JGL10}, with
applications do the verification of infinite state systems such as orbit
finite weighted automata \cite{BOKLMO21}, \kl{orbit finite polynomial
automata}, and more generally orbit finite systems dealing with polynomial
computations.

\paragraph{Organisation.} \AP The rest of the paper is organised as follows. In
\cref{sec:preliminaries}, we introduce formally the notions of \kl{Gröbner
bases}, \kl{effectively oligomorphic} actions, and \kl{well-quasi-orderings},
which are the main assumptions of our positive results. Then, we illustrate in
\cref{sec:act ex} how these assumptions can be satisfied in practice, providing
numerous examples of sets of indeterminates. After that, we introduce in
\cref{sec:weakgb} an adaptation of \kl{Buchberger's algorithm} to the
equivariant case, that computes a \kl{weak equivariant Gröbner basis} of an
\kl{equivariant ideal}. In \cref{sec:equivariant-grobner-basis}, we use
\kl{weak equivariant Gröbner bases} to prove our main positive
\cref{thm:compute-egb}, and we show that it provides a way to effectively
represent \kl{equivariant ideals} (\cref{cor:equivariant-ideals-computations}).
We continue by showing in \cref{sec:closure-properties} that the assumptions of our
\cref{thm:compute-egb} are closed under two natural operations
(\cref{lem:closure-properties-comp,lem:closure-properties-wqo}). The positive
results regarding the \kl{equivariant ideal membership problem} are then
leveraged to obtain several decision procedures for other problems in
\cref{sec:applications}. Finally, in \cref{sec:undecidability}, we show that
our assumptions are close to optimal by proving that the \kl{equivariant ideal
membership problem} is undecidable whenever one can find \kl(of){infinite
paths} in the set of indeterminates (\cref{thm:undecidable-paths}), which is
conjectured to be a complete characterisation of the undecidability of the
\kl{equivariant ideal membership problem} (\cref{rem:conj-wqo-infinite-path}).
\section{Preliminaries}
\label{sec:preliminaries}

\paragraph{Partial orders, ordinals, well-founded sets, and well-quasi-ordered
sets.} \AP We assume basic familiarity with partial orders, well-founded sets,
and ordinals. We will use the notation $\intro*\om$ for the first infinite
ordinal (that is, $(\N, \leq)$), and write $X \intro*\ordplus Y$ for the
lexicographic sum of two partial orders $X$ and $Y$. Similarly, the notation $X
\intro*\ordtimes Y$ will denote the product of two partial orders equipped with the
lexigographic ordering, i.e. $(x_1, y_1) \leq (x_2, y_2)$ if either $x_1 <
x_2$, or $x_1 = x_2$ and $y_1 \leq y_2$. We will also use the usual notations
for finite ordinals, writing $\intro*\ordfin{n}$ for the finite ordinal of size
$n$. For instance, $\om \ordplus \ordfin{1}$ is the total order $\N \uplus
\set{+\infty}$, where $+\infty$ is the new largest element.

\AP In order to guarantee the termination of the algorithms presented in this
paper, a key ingredient will be the notion of \intro{well-quasi-ordering}
(WQO), that are sets $(X, \leq)$ such that every infinite sequence
$\seqof{x_i}[i \in \N]$ of elements of $X$ contains a pair $i < j$ such that
$x_i \leq x_j$. Examples of \kl{well-quasi-orderings} include finite sets with
any ordering, or $\N \times \N$ with the product ordering. We refer the reader
to \cite{SCSC17} for a comprehensive introduction to \kl{well-quasi-orderings}
and their applications in computer science.

\paragraph*{Polynomials, monomials, divisibility.} \AP 
We assume basic familiarity with the theory of
commutative algebra, and polynomials. We will use the notation $\poly{\K}{\Indets}$
for the ring of polynomials with coefficients from a field $\K$ and
indeterminates/variables from a set $\Indets$, and $\mon{\Indets}$ for the set of
monomials in $\poly{\K}{\Indets}$. Letters $p,q,r$ are used to denote polynomials,
$\monelt,\monelt[n]$ are used to denote monomials, and $a,b,\alpha,\beta$ are
used to denote coefficients in $\K$.

A classical example of a \kl{WQO} is the set of monomials $\mon{\Indets}$,
endowed with the \kl{divisibility} relation $\intro*\divleq$ whenever $\Indets$
is finite. We recall that a monomial $\monelt[m]$ \intro{divides} a monomial
$\monelt[n]$ if there exists a monomial $\monelt[l]$ such that $\monelt[m]
\times \monelt[l] = \monelt[n]$. In this case, we write $\monelt[m]
\reintro*\divleq \monelt[n]$. Note that monomials can be seen as functions from
$\Indets$ to $\N$ having a finite support, and that the \kl{divisibility}
relation can be extended to monomials that are functions from $\Indets$ to
$(Y,\leq)$, where $Y$ is any partially ordered set. In this case, we write
$\monelt[m] \divleq \monelt[n]$ if for every $x \in \Indets$, we have
$\monelt[m](x) \leq \monelt[n](x)$. We will write $\intro*\mon[\om \ordplus
1]{\Indets}$ (resp. $\mon[\om^2]{\Indets}$) for the set of monomials that are
functions from $\Indets$ to $\om \ordplus \ordfin{1}$ (resp. $\om^2$).

\AP Unless otherwise specified, we will assume that the set of indeterminates
$\Indets$ comes equipped with a total ordering $\varleq$. Using this order, we
define the \intro{reverse lexicographic} (revlex) ordering on monomials as
follows: $\monelt[n] \intro*\revlexlt \monelt[m]$ if there exists an
indeterminate $x \in \Indets$ such that $\monelt[n](x) < \monelt[m](x)$, and such
that for every $y \in \Indets$, if $x \varlt y$ then $\monelt[n](y) =
\monelt[m](y)$. Remark that if $\monelt[n] \revlexleq \monelt[m]$, then in
particular $\monelt[n] \divleq \monelt[m]$. 

\AP We can now use the \kl{reverse lexicographic} ordering to identify particular elements in
a given polynomial. Namely, for a polynomial $p \in \poly{\K}{\X}$, we define
the \intro{leading monomial} $\intro*\lm(p)$ of $p$ as the largest monomial
appearing in $p$ with respect to the \kl{revlex ordering}, and the
\intro{leading coefficient} $\intro*\lc(p)$ of $p$ as the coefficient of
$\lm(p)$ in $p$. We can then define the \intro{leading term} $\intro*\lt(p)$ of
$p$ as the product of its \kl{leading monomial} and its \kl{leading
coefficient}, and the \intro{characteristic monomial} $\intro*\cm(p)$ of $p$ as
the product of its \kl{leading monomial} and all the indeterminates appearing
in $p$. We also define the \intro(monomial){domain} of $\monelt[m]$ as the set
$\intro*\dom(\monelt[m])$ of indeterminates $x \in \X$ such that $\monelt[m](x) \neq
0$. Because the coefficients and monomial in question are highly dependent on
the ordering $\varleq$, we allow ourselves to write $\lm[\Indets](p)$ to
highlight the precise ordered set of variables that was used to compute the
\kl{leading monomial} of $p$.

Let us briefly argue in favor of the \kl{reverse lexicographic} ordering. In
the case of a finite set of indeterminates, one can choose any total ordering
on $\mon{\Indets}$, as long as it contains the \kl{divisibility}
quasi-ordering, and is compatible with the product of monomials.\footnote{This
is often called a \emph{monomial ordering}, see \cite{CLO15}.} In our case,
having an infinite number of indeterminates, we rely on a connection between
$\lm(p)$ and $\dom(p)$: $\dom(p) \subseteq \dwset{\dom(\lm(p))}$, where
$\dwset{S}$ is the downward closure of a set $S \subseteq \Indets$, i.e. the
set of all indeterminates $x \in \Indets$ such that $y \leq x$ for some $y \in
S$. This means that the \kl{leading monomial} encodes a \emph{global property}
of the polynomial, and it will be crucial in our termination arguments. This is
already at the core of the classical \emph{elimination theorems} \cite[Chapter 3, Theorem
2]{CLO15}.

\paragraph{Ideals, and Gröbner Bases.} \AP An \intro{ideal} $\idl$ of
$\poly{\K}{\X}$ is a non-empty subset of $\poly{\K}{\X}$ that is closed under
addition and multiplication by elements of $\poly{\K}{\X}$. Given a set $H
\subseteq \poly{\K}{\X}$, we denote by $\intro*\IdlGen{H}$ the ideal generated
by $H$, i.e. the smallest ideal that contains $H$. The \intro{ideal membership
problem} is the following decision problem: given a polynomial $p \in
\poly{\K}{\X}$ and a set of polynomials $H \subseteq \poly{\K}{\X}$, decide
whether $p$ belongs to the ideal $\IdlGen{H}$ generated by $H$. We know that
this problem is decidable when $\X$ is finite, and that it is even
$\EXPTIME$-complete \cite{MAME82}. The classical approach to the \kl{ideal
membership problem} is to use the \kl{Gröbner basis} theory that was developed
in the 70s by Buchberger~\cite{BUCH76}. 
A set $\Basis$ of polynomials is called a \intro{Gröbner basis} of
an ideal $\idl$ if, $\IdlGen{\Basis} = \idl$ and for every polynomial $p \in
\idl$, there exists a polynomial $q \in \Basis$ such that $\lm[\Indets](q)
\divleq \lm[\Indets](p)$.

Given a \kl{Gröbner basis} $\Basis$ of an ideal $\idl$, and a polynomial $p$,
it suffices to iteratively reduce the \kl{leading monomial} of $p$ by
subtracting multiples of elements in $\Basis$, until one cannot apply any
reductions. If the result is $0$, then $p$ belongs to $\idl$, and otherwise it
does not.

\paragraph{Group actions, equivariance, and orbit finite sets.}  \AP A
\intro{group} $\group$ is a set equipped with a binary operation that is
associative, has an identity element and has inverses. In our setting, we are
interested in infinite sets $\X$ of indeterminates that is equipped with a
\intro{group action} $\group \actson \X$. This means that for each $\gelem \in
\group$, we have a bijection $\X \tobij \X$ that we denote by $x \mapsto \gelem
\cdot x$. A set $S \subseteq \X$ is \intro{equivariant} under the action of
$\group$ if for all $\gelem \in \group$ and $x \in S$, we have $\gelem \cdot x
\in S$. We give in \cref{ex:idl-equiv} an example and a non-example
of \intro{equivariant ideals}.

\begin{example}
    \label{ex:idl-equiv}
    Let $\Indets$ be any infinite set, and $\group$ be the 
    group of all bijections of $\Indets$. 
    Then the set $S_0 \subset \poly{\K}{\Indets}$ of all polynomials 
    whose set of coefficients sums to $0$ is an equivariant ideal.
    Conversely, the set of all polynomials that are multiple
    of $x \in X$ is an \kl{ideal} that is not \kl{equivariant}.
\end{example}
\begin{proof}
    Let $p,q\in S_0$, and $r \in \poly{\K}{\Indets}$.
    Then, $p \times r + q$ is in $S_0$. Remark that 
    $p,r$ and $q$ belong to a subset $\poly{\K}{\Indets}$ of the 
    polynomials that uses only finitely many indeterminates.
    In this subset, the sum of all coefficients is obtained
    by applying the polynomials to the value $1$ for every indeterminate
    $y \in \Indets$. We conclude that
    $(p \times r + q)(1,\dots, 1) 
    = p(1,\dots,1) \times r(1,\dots,1) + q(1,\dots,1)
    = 0 \times r(1, \dots, 1) + 0 = 0$, hence that
    $p \times r + q$ belongs to $S_0$. 
    Because $0$ is in $S_0$, we conclude that $S_0$ is an \kl{ideal}.
    Furthermore, if $\gelem \in \group$ and $p \in S_0$, then
    the sum of the coefficients $\gelem \cdot p$ is exactly
    the sum of the coefficients of $p$, hence is $0$ too.
    This shows that $S_0$ is \kl(ideal){equivariant}.

    It is clear that all multiples of a given polynomial $x \in \Indets$
    is an ideal of $\poly{\K}{\Indets}$. This is not an \kl{equivariant ideal}:
    take any bijection $\gelem \in \group$ that does not map $x$ to $x$ (it
    exists because $\Indets$ is infinite and $\group$ is all permutations),
    then $\gelem \cdot x$ is not a multiple of $x$, and therefore does 
    not belong to the ideal.
\end{proof}

\AP An \kl{equivariant set} is said to be \intro{orbit finite} if it is the
union of finitely many \intro{orbits} under the action of $\group$. We denote
$\intro*\orbit[\group]{E}$ for the set of all elements $\gelem \cdot x$ for
$\gelem \in \group$ and $x \in E$. Equivalently, an \reintro{orbit finite set}
is a set of the form $\orbit[\group]{E}$ for some finite set $E$. Not every
\kl{equivariant subset} is \kl{orbit finite}, as shown in
\cref{ex:orbit-finite}. However, \kl{orbit finite sets} are
robust in the sense that \kl{equivariant subsets} of \kl{orbit finite sets} are
also \kl{orbit finite}, and similarly, an \kl{equivariant subset} of $E^n$ is
\kl{orbit finite} whenever $E$ is \kl{orbit finite} and $n \in \N$ is finite.
For algorithmic purposes, \kl{orbit finite sets} are the ones that can be taken
as input as a finite set of representatives (one for each orbit). The notions
of \kl{equivariance} and \kl{orbit finite sets} from a computational
perspective are discussed in \cite{BOJAN16inf}, and we refer the reader to this
book for a more comprehensive introduction to the topic.

\AP We will mostly be interested in \intro{orbit-finitely generated}
\kl{equivariant ideals}, i.e.\ \kl{equivariant ideals} that are generated by an
\kl{orbit finite set} of polynomials, for which the \intro{equivariant ideal
membership problem} is as follows: given a polynomial $p \in
\poly{\K}{\Indets}$ and an \kl{orbit finite set} $H \subseteq
\poly{\K}{\Indets}$, decide whether $p$ belongs to the \kl{equivariant ideal}
$\intro*\EqIdlGen{H}$ generated by $H$.

\begin{example}
  \label{ex:orbit-finite}
  Let $\Indets = \N$, and $\group$ be all permutations 
  that fixes prime numbers. The
  set of all polynomials whose coefficients sum to $0$ is an 
  \kl{equivariant ideal}, but it is not \kl{orbit finite},
  since all the polynomials $x_p - x_q$ for $p \neq q$ primes
  are in distinct orbits under the action of $\group$.
\end{example}

\AP A function $f \colon X \to Y$ between two sets $X$ and $Y$ equipped with
actions $\group \actson X$ and $\group \actson Y$ is said to be
\intro(func){equivariant} if for all $\gelem \in \group$ and $x \in X$, we have
$f(\gelem \cdot x) = \gelem \cdot f(x)$. For instance, the
\kl(monomial){domain} of a monomial is an \kl{equivariant function} if $\gelem
\in \group$, then $\gelem \cdot \dom(\monelt[m]) = \dom(\gelem \cdot
\monelt[m])$. Let us point out that the image of an \kl{orbit finite set} under
an \kl{equivariant function} is \kl{orbit finite}, and that the algorithms that
we will develop in this paper will all be \kl(func){equivariant}.

\paragraph*{Computability assumptions.} \AP We say that the action is
\intro{effectively oligomorphic} if :
\begin{enumerate}
\item It is \intro{oligomorphic}, i.e.\ for every $n \in \N$ and every \kl{orbit
finite set} $E \subseteq \Indets$,
the set $E^n$ is \kl{orbit finite} under the action of $\group$ on $\Indets^n$.
\item There exists an algorithm that decides whether two elements $\vec{x},
\vec{y} \in \Indets^*$ are in the same orbit under the action of $\group$ on $\Indets^*$.
\item There exists an algorithm which on input $n\in\N$ outputs a set $A\subseteqfin\Indets^n$ such that $|A\cap U| = 1$ for every orbit $U\in\Indets^n$.
\end{enumerate}

A group action $\group \actson \X$ is said to be \intro(ord){compatible}
with an ordering $\leq$ on $\X$ if for all $\gelem \in \group$ and $x,y \in
\X$, we have $x \leq y$ if and only if $\gelem \cdot x \leq \gelem \cdot y$.
Let us point out that in this case, $\revlexleq$ is also \kl(ord){compatible} with
the action of $\group$ on $\mon{\X}$, i.e. for all $\gelem \in \group$ and
monomials $\monelt[m], \monelt[n] \in \mon{\X}$, we have $\monelt[m] \revlexleq
\monelt[n]$ if and only if $\gelem \cdot \monelt[m] \revlexleq \gelem \cdot
\monelt[n]$.
Our \intro{computability assumptions} on the tuple $(\Indets, \group,
\leq)$ will therefore be that $\group$ acts \kl{effectively oligomorphic} on
$\Indets$, and that its action is \kl(ord){compatible} with the ordering $\leq$
on $\Indets$.

\begin{example}
  \label{ex:computability-assumptions}
  Let $\Indets \defined \Q$ and $\group$ be the group of all
  order preserving bijections of $\Q$.
  Then, $\group$ acts \kl{effectively oligomorphically} on $\Indets$,
  and its action is \kl(ord){compatible} with the ordering of $\Q$ by definition.
\end{example}

Note that under our \kl{computability assumptions}, the set of polynomials
$\poly{\K}{\Indets}$ is also \kl{effectively oligomorphic} under the action of
$\group$ on $\Indets$. This is because a polynomial $p \in \poly{\K}{\Indets}$
can be seen as an element of $(\K \times \Indets^{\leq d})^n$ where $n$ is the
number of monomials in $p$, and $d$ is the maximal degree of a monomial
appearing in $p$. Beware that the set of polynomials $\poly{\K}{\Indets}$ is
not \kl{orbit finite}, precisely because the orbit of a polynomial $p$ under 
the action of $\group$ cannot change the degree of $p$, and that there are 
polynomials of arbitrarily large degree.

\paragraph{Equivariant Gröbner bases.} \AP We know from \cite{GHOLAS24} that a
necessary condition for the \kl{equivariant Hilbert basis property} to hold is
that the set  $\mon{\X}$  of monomials is a \kl{well-quasi-ordering} when
endowed with the \intro{divisibility up-to $\group$} relation
($\intro*\gdivleq$), which is defined as follows: for $\monelt_1, \monelt_2 \in
\mon{\X}$, we write $\monelt_1 \gdivleq \monelt_2$ if there exists $\gelem \in
\group$ such that $\monelt_1$ \kl{divides} $\gelem \cdot \monelt_2$. This
relation also extends to monomials that are functions from $\Indets$ to
$(Y,\leq)$ with finite support, where $Y$ is any partially ordered set. We say
that a set $\Basis \subseteq \poly{\K}{\X}$ is an \intro{equivariant Gr\"{o}bner
basis} of an equivariant ideal $\idl$ if $\Basis$ is \kl{equivariant},
$\IdlGen{\Basis} = \idl$, and for every polynomial $p \in \idl$, there exists
$q \in \Basis$ such that $\lm[\Indets](q) \gdivleq \lm[\Indets](p)$ and
$\dom(q) \subseteq \dom(p)$, following the definition of \cite{GHOLAS24}.

Beware that even in the case of a finite set of variables, a \kl{Gröbner basis}
is not necessarily an \kl{equivariant Gröbner basis}, because of the
\kl(polynomial){domain} condition. However, every \kl{equivariant Gröbner
basis} is a \kl{Gröbner basis}.

\begin{example}
  \label{ex:equivariant-gb}
  Let $\Indets \defined \set{ x_1, x_2 }$,
  with $x_1 \varleq x_2$,
  and $\group$ be the trivial group.
  Let us furthermore consider the ideal $\idl$ \kl(idl){generated by}
  $\set{ x_1, x_2 }$.
  Then, the set $\Basis \defined \set{ x_2 - x_1, x_1 }$ is a
  \kl{Gröbner basis} of $\idl$, but not an \kl{equivariant Gröbner basis}.
  Indeed, $x_2 \in \idl$, but there is no polynomial $q \in \Basis$
  such that $\lm(q) \divleq x_2$ and $\dom(q) \subseteq \dom(x_2)$.
\end{example}

In the finite case, one can always compute an \kl{equivariant Gröbner basis} by
computing \kl{Gröbner bases} for every possible ordering of the indeterminates,
and taking their union.\footnote{This algorithm is correct because we are
  considering the \kl{reverse lexicographic} ordering.}

\section{Examples of group actions}\label{sec:act ex}

\AP Many of the common examples of group actions $\group\actson\Indets$ are
obtained by considering $\Indets$  as set with some structure, described by
some relations and functions on that set, and $\group$ is the group
$\aut{\Indets}$ of all automorphisms (i.e.\ bijections that preserve and
reflect the structure) of $\Indets$. A monomial $\monelt[p] \in
\mon[Q]{\Indets}$ can be thought as a labelling of a finite substructure of
$\X$ using elements of $Y$. If the structure $\Indets$ is \intro{homogeneous},
that is, if isomorphisms between finite induced substructures extends to
automorphisms of the whole structure, then $\gdivleq$ is same as embedding of
labelled finite induced substructures of $\X$. \footnote{ We refer the reader
to \cite[Chapter 7]{BOJAN16inf} and \cite{homsurvey} for more details on
homogeneous structures.} Let us now give some examples of such structures and
whether they satisfy our \kl{computability assumptions}, and whether the
\kl{divisibility relation up-to-$\group$} is a \kl{well-quasi-ordering}.

\begin{example}[Equality Atoms]
  \label{ex:eq atoms}
  \AP
Let $\intro*\EqualityAtoms$ be an infinite set without any additional structure other than the equality relation.
Up to isomorphism, finite induced substructures of $\EqualityAtoms$ are finite sets,
monomials in $\mon[Y]{\EqualityAtoms}$ are finite multisets of elements in $Y$,
and $\gdivleq[\aut{\EqualityAtoms}]$ is the multiset ordering \cite[Section 1.5]{SCSC17},
which is a \kl{WQO} \cite[Corollary 1.21]{SCSC17}.
\end{example}

\begin{example}[Dense linear order]
  \label{ex:dlo}
  \AP
Let $\intro*\OrderAtoms$ be the set of rational numbers ordered by the usual ordering.
Note that under this ordering, $\OrderAtoms$ is a dense linear order without endpoints.
We write $\OrderAtoms$ instead of $\Q$ to emphasise that we use its elements as indeterminates and not as coefficients of polynomials. 
Up to isomorphism, finite induced substructures of $\OrderAtoms$ are finite linear orders,
monomials in $\mon[Y]{\OrderAtoms}$ are words in $Y^*$ (i.e.\ finite linear order labelled with elements of $Y$)
and $\gdivleq[\aut{\OrderAtoms}]$ is the scattered subword ordering, which is a \kl{WQO} due to Higman's lemma \cite{HIG52}.
\end{example}

\begin{example}[The Rado graph]
  \label{ex:rado}
  \AP
  Let $\intro*\RadoAtoms$ be the \intro{Rado graph} (\cite[Section 7.3.1]{BOJAN16inf},\cite[Example 2.2.1]{homsurvey}).
Up to isomorphism,
finite induced substructures of $\RadoAtoms$ are finite undirected graphs,
monomials in $\mon[Y]{\RadoAtoms}$ are graphs with vertices labelled with $Y$,
and $\gdivleq[\aut{\RadoAtoms}]$ is the labelled induced subgraph ordering even when $Y$ is a singleton.
For example, cycles of length more than three form an infinite antichain.
\end{example}

\begin{example}[Infinite dimensional vector space]
  \label{ex:bit vector}
  \AP
Let $\intro*\BitVectorAtoms$ be an infinite dimensional vector space over $\ftwo$.
Up to isomorphism,
finite induced substructures of $\BitVectorAtoms$ are finite dimensional vector spaces over $\ftwo$.
These are well-quasi-ordered in the absence of labelling.
However, even when $Y = \N$,
$(\mon[Y]{\BitVectorAtoms},\gdivleq[\aut{\BitVectorAtoms}])$ is not a \kl{WQO} as illustrated by the following antichain.
Let $\{v_1,v_2,\dots\}\subseteq \BitVectorAtoms$ be a countable set of linearly independent vectors in $\BitVectorAtoms$.
Let $\oplus$ denote the addition operation of $\BitVectorAtoms$.
For $n \geq 3$ define the monomial 
$
\monelt[p]_n \defined v^2_1 \ldots v^2_n  (v_1\oplus v_2)(v_2\oplus v_3) \ldots (v_{n-1}\oplus v_n)  (v_{n}\oplus v_1)
$.
Then, $\setof{\monelt[p]_n}{n = 3,4,\dots}$ forms an infinite antichain.
\end{example}

The previous \cref{ex:eq atoms,ex:dlo,ex:rado,ex:bit vector}
are
well known examples in the theory of \emph{sets with atoms} \cite{BOJAN16inf}.
Let us now give a new example of \kl{well-quasi-ordered} \kl{divisibility relation
up-to-$\group$}, by extending \Cref{ex:dlo} that
relied on Higman's lemma \cite{HIG52} via Kruskal's tree theorem
\cite{Kruskal60}.

\begin{example}[Dense Tree]
  \label{ex:dense tree}
  \AP
Let $\intro*\TreeAtoms$ denote the universal countable dense meet-tree, as
defined in 
\cite[Page 2]{KRS21} or \cite[Section 7.3.3]{BOJAN16inf}.
Note that the tree structure is given by the \emph{least common ancestor} (\emph{meet})
operation, and not by its edges.
For a subset $S\subset \TreeAtoms$,
define its \emph{closure} to be the smallest subtree of $\TreeAtoms$ containing $S$.
Up to isomorphism, finite induced substructures of $\TreeAtoms$ are finite meet-trees.
Monomials in $\mon[Y]{\TreeAtoms}$ are finite meet-trees labelled with $1 + Y$.
Here $1 + Y$ is the \kl{WQO} containing one more element than $Y$ which is incomparable to elements in $Y$,
and is used to label nodes that are in the closure of the set of variable of a monomial, but not in the monomial itself.
The divisibility relation $\gdivleq[\aut{\TreeAtoms}]$ is exactly the embedding of labelled meet-trees,
which is a \kl{WQO} due to Kruskal's tree theorem \cite{Kruskal60}.
\end{example}

The above examples using \kl{homogeneous} structures nicely illustrate the
correspondence between monomials and labelled finite substructures, but we can
also consider \kl{non-homogeneous} structures, such as in \Cref{ex:int} below.

\begin{example}\label{ex:int}
  \AP
Let $\calZ$ be the set of integers ordered by the usual ordering.
Then $\aut{\calZ}$ is the set of all order preserving bijections of $\D$.
Note that every order preserving bijection of the set $\calZ$ is a translation $n \mapsto n + c$ for some constant $c\in\calZ$.
By definition, the action $\aut{\calZ} \actson \calZ$ preserves the linear order on $\Z$.
However, $(\mon[Y]{\calZ}, \gdivleq[\aut{\calZ}])$ is not a \kl{WQO} even when $Y$ is a singleton.
An example of an infinite antichain is the set $\setof{a b}{b\in\calZ\setminus\{a\}}$, for any fixed $a\in\calZ$.
\end{example}

\paragraph{Reducts} \AP As mentioned in the introduction, some examples of
group actions $\group\actson\Indets$ do not preserve a linear order on
$\Indets$, such as the set $\EqualityAtoms$ of indistinguishable names without
any relations. However, there are techniques that allow us to reduce the
problem of computing \kl{equivariant Gröbner bases} to the case where the
action preserves a linear ordering, using what is called a \kl{reduct} of the
action. A group action $\group\actson\Indets$ is said to be a \intro{reduct} of
another group action $\calH\actson\Y$ if there exists a bijection $f \colon
\Indets\to\Y$ such that $f^{-1}\circ \gelem \circ f \in\group$ for every
$\gelem \in \calH$. It is called an \intro{effective reduct} if $f$ is
computable.

\begin{lemma}\label{lem:reducts-equiv-hilbert}
  Let $\group\actson\Indets$ be an \kl{effective reduct} of $\calH\actson\Y$ 
  via a computable function $f\colon\Indets\to\Y$. Then,
  \begin{enumerate}
    \item for every $\group$-equivariant ideal $\idl[I]\subseteq\poly{\K}{\X}$,
    $f(\idl[I])\subseteq\poly{\K}{\Y}$ is a $\calH$-equivariant ideal,
    \item if $\Basis\subseteq\poly{\K}{\Y}$ is a $\calH$-equivariant Gr\"{o}bner basis then $\orbit[\group]{f^{-1}(\Basis)}$ is a $\group$-equivariant Gr\"{o}bner basis.
  \end{enumerate}
\end{lemma}

\AP We say that an action $\group\actson\Indets$ is \intro{nicely orderable} if
it is an \kl{effective reduct} of some  \kl{effectively oligomorphic} action
$\calH\actson\Y$ that satisfies the conditions mentioned in
\Cref{thm:compute-egb},
i.e.\ $\calH\actson\Y$ preserves an effective linear order and
$\gdivleq[\calH]$ is a \kl{WQO}. As an immediate corollary of
\Cref{thm:compute-egb},
\Cref{lem:reducts-equiv-hilbert} we get the desired \cref{cor:egb}.
  
\begin{corollary}\label{cor:egb} If $\group\actson\Indets$ is
  \kl{nicely orderable} then
  one can compute an \kl{equivariant Gröbner basis} of any 
  \kl{equivariant ideal} given by an \kl{orbit finite} set of
  generators.
\end{corollary}

\begin{remark}\label{rem:reduct}

\Cref{lem:reducts-equiv-hilbert} implies that one can apply our
results to an action $\group\actson\Indets$ that does not preserve a linear
order, as soon as it is  a reduct of some another action $\calH\actson\Indets$ which
does preserves a linear order. 

For example, $\aut{\EqualityAtoms}\actson\EqualityAtoms$ is a \kl{reduct} of
$\aut{\OrderAtoms}\actson\OrderAtoms$ assuming $\EqualityAtoms$ is countable.
Similarly, let $\T_<$ be the countable dense-meet tree with a lexicographic
ordering, as defined in \cite[Remark 6.14]{KRS21}.\footnote{The remark says
that finite meet-trees expanded with a lexicographic ordering is a Fra\"{i}sse
class, from which it follows that there exists a Fra\"{i}sse limit $\T_<$ for
that class.} Let $\group$ be the group of bijections of $\T_<$ which do not
necessarily preserve the lexicographic ordering. Then $\group\actson\T_<$ is
isomorphic to $\aut{\TreeAtoms}\actson\TreeAtoms$, and hence
$\aut{\TreeAtoms}\actson\TreeAtoms$ is a reduct of $\aut{\T_<}\actson\T_<$.
\end{remark}

\paragraph{Effective oligomorphicity} \AP Recall that in our \kl{computability
assumptions} we require the action $\group\actson\Indets$ to be \kl{effectively
oligomorphic}. It is already known that all the structures of \cref{ex:eq
atoms,ex:dlo,ex:rado,ex:bit vector,ex:dense tree} are \kl{oligomorphic}
\cite[Theorem 7.6]{BOJAN16inf}. 

Let us show on an example that they are also
\kl{effectively oligomorphic}. It is clear that $\OrderAtoms$ can be
represented by integer fractions, and that the \kl{orbit} of a tuple $(q_1,
q_2, \ldots, q_n)$ of rational numbers is given by their relative ordering in
$\Q$, which can be effectively computed. Finally, one can enumerate such
orderings and produce representatives by selecting $n$ integers. 
This can be generalised to all the structures mentioned in \Cref{ex:eq
atoms,ex:dlo,ex:rado,ex:bit vector,ex:dense tree}, by using dedicated
representations (such as \cite[Page 244-245]{BOJAN16inf} for $\TreeAtoms$), or
the general theory of Fra\"{i}sse limits \cite{CompFraisse}. 

Finally, let us
remark that \cref{rem:reduct} preserves the \kl{effective oligomorphicity} of
the action.
\section{Weak Equivariant Gröbner Bases}
\label{sec:weakgb}

\AP In this section we prove that a natural adaptation of \kl{Buchberger's
algorithm} to the equivariant setting computes a \kl{weak equivariant Gröbner
basis} of an \kl{equivariant ideal}. This can be seen as an analysis of the
classical algorithm in the equivariant setting. We will assume for the rest of
the section that $\Indets$ is a set of indeterminates equipped with a group
$\group$ acting \kl{effectively oligomorphically} on $\X$, and that $\X$ is
equipped with a total ordering $\varleq$ that is \kl(ord){compatible} with the
action of $\group$. The crucial object of this section is the notion of
\kl{decomposition} of a polynomial with respect to a set $H$.

\begin{definition}
  \label{def:decomposition}
  Let $H$ be a set of polynomials. A \intro{decomposition} of $p$
  with respect to $H$ is given by a finite sequence 
  $\mathfrak{d} \defined \seqof{(a_i, \monelt_i, h_i)}[i \in I]$ such that
   $ p = \sum_{i \in I} a_i \monelt_i h_i$,
  where $a_i \in \K$, $\monelt_i \in \mon{\X}$, and $h_i \in H$ for all $i \in I$.
  The \intro{domain of the decomposition}
  that we write $\intro*\domdec(\mathfrak{d})$ is defined as the union
  of the domains of the polynomials $\monelt_i h_i$ for all $i \in I$.
  The \intro{leading monomial of the decomposition} is defined as
  $
    \intro*\lmdec(\mathfrak{d}) \defined \max(\seqof{\lm(\monelt_i h_i)}[i \in I])
  $.
\end{definition}

Leveraging the notion of decomposition, we can define a weakening of the notion
of \kl{equivariant Gröbner basis}, that essentially mimics the classical notion
of \kl{equivariant Gröbner basis} at the level of \kl{decompositions} instead
of polynomials.

\begin{definition}
  An \kl{equivariant set} $\Basis$ of polynomials is 
  a \intro{weak equivariant Gröbner basis} of an \kl{equivariant ideal}
  $\idl$ if $\IdlGen{\Basis} = \idl$, and if for every polynomial $p \in \idl$,
  and decomposition $\mathfrak{d}$ of $p$ with respect to $\Basis$, there
  exists a decomposition $\mathfrak{d}'$ of $p$ with respect to $\Basis$ such that
  $\domdec(\mathfrak{d}') \subseteq \domdec(\mathfrak{d})$,
  and 
  such that $\lmdec(\mathfrak{d}') = \lm(p)$.
\end{definition}

\AP To compute \kl{weak equivariant Gröbner bases}, we will use a rewriting
relation. Given $p,r \in \poly{\K}{\X}$, we write $p \intro*\toeucl{H}
r$ if and only if there exists $q \in H$, $a \in \K$, and $\monelt \in
\mon{\X}$ such that $p = a \monelt q + r$, $\dom(r) \subseteq \dom(p)$, and
$\lm[\Indets](r) \revlexlt \lm[\Indets](p)$. In order to simplify the
notations, we will write $p \intro*\pmonlt r$ to denote $\dom(r) \subseteq
\dom(p)$, and $\lm[\Indets](r) \revlexlt \lm[\Indets](p)$, leaving the
ordered set of indeterminates $\Indets$ implicit.
The relation $\pmonleq$ is extended to decompositions by using 
the analogues of $\dom$ and $\lm$ for decompositions.

\begin{lemma}
  \label{lem:chm}
  The quasi-ordering $\pmonleq$ is \kl(ord){compatible} with the action of $\group$,
  and is well-founded on polynomials, and on \kl{decompositions} of polynomials.
\end{lemma}
\begin{proof}
  The first property is immediate because $\dom$, $\lm$, and $\revlexleq$ are
  compatible with the group action $\group$. 
  The second property follows from the fact that $\revlexlt$ is a total
  well-founded ordering whenever one has fixed finitely many possible 
  indeterminates. In a decreasing sequence, the support of the leading 
  monomials is also decreasing, so that sequence only contains finitely many 
  indeterminates, hence we conclude.
  The same proof works for decompositions.
\end{proof}

\AP As a consequence of \cref{lem:chm}, we know that the rewriting relation
$\toeucl{H}$ is \intro{terminating} for every set $H$. Given a set $H$ of
polynomials, and given a polynomial $p \in \poly{\K}{\X}$, we say that $p$ is
\intro{normalised} with respect to $H$ if there are no transitions $p
\toeucl{H} r$. The set of \intro{remainders} of $p$ with respect to $H$ is
denoted $\intro*\rem{H}{p}$, and is defined as the set of all polynomials $r$
such that $p \toeucl{H}^* r$ and $r$ is \kl{normalised} with respect to $H$.
The following lemma states that $\rem{H}{\cdot}$ is a computable function from
in our setting.

\begin{lemma}
  \label{lem:normalisation}
  Let $H$ be an \kl{orbit finite set} of polynomials, and let $p \in \poly{\K}{\X}$ be a
  polynomial. Then $\rem{H}{p}$ is finite.
  Furthermore, this computation
  is \kl(func){equivariant}. In particular, 
  $\rem{H}{K}$ is a computable \kl{orbit finite} for every \kl{orbit finite} set $K$ of polynomials.
\end{lemma}
\begin{proof}
  Let us write $H = \orbit[\group]{H'}$, where $H'$ is a finite set of
  polynomials.
  Because the relation $\toeucl{H}$ is \kl{terminating}, it suffices to 
  show that for every polynomial $p$, there are finitely many polynomials $r$ 
  such that $p \toeucl{H} r$, leveraging König's lemma. This is because 
  $p \toeucl{H} r$ implies that 
  $p = \alpha \monelt[n] (\gelem\cdot q) + r$ for some $q \in H'$, 
  $\alpha \in \K$, $\monelt[n] \in \mon{\X}$, and $\gelem \in \group$.
  Because, $\lm(r) \revlexlt \lm(q)$, we  
  conclude that $\lm(p) = \lm(\alpha \monelt[n] (\gelem\cdot q))$, and 
  therefore $r$ is uniquely determined by the choice of $q \in H'$ and the
  choice of $\gelem \in \group$ that maps the \kl(poly){domain} of $q$ to the \kl(poly){domain} of
  $p$. There are finitely elements in $H'$ and finitely many such functions
  $\gelem \in \group$ because both domains are finite.
\end{proof}

\AP Now that we have a quasi-ordering on polynomials, we will prove that given
an \kl{orbit finite} set $H$ of generators, we can compute a \kl{weak
equivariant Gröbner basis}. The computation will closely follow the classical
\kl{Buchberger's algorithm}. The main idea being to saturate the set of
generators $H$ to remove some \emph{critical pairs} of the rewriting relation
$\toeucl{H}$. Namely, given two polynomials $p$ and $q$ in $H$, we compute the
set $\intro*\CancelPoly{p}{q}$ of cancellations between $p$ and $q$ as the set of
polynomials of the form $r = \alpha \monelt[n] p + \beta \monelt[m] q$ such
that $\lm(r) < \max(\monelt[n] \lm(p), \monelt[m]\lm(q))$, where $\alpha,\beta
\in \K$, and where $\monelt[n], \monelt[m] \in \mon{\X}$. Let us recall that
given two monomials $\monelt[n], \monelt[m] \in \mon{\X}$, one can compute
$\intro*\lcm(\monelt[n], \monelt[m])$ as the least common multiple of the two
monomials, and that this in an \kl{equivariant operation}.
Using this, we can introduce the \intro{S-polynomial} of two polynomials $p$ and $q$
as in \cref{eq:spoly}.
 \begin{equation}
    \label{eq:spoly}
    \intro*\spoly{p}{q} \defined
    \frac{\lcm(\lm(p), \lm(q))}{\lt(p)} \times p
    - \frac{\lcm(\lm(p), \lm(q))}{\lt(q)} \times q
    \quad .
  \end{equation}

\begin{lemma}
  \label{lem:spoly}
  Let $p$ and $q$ be two polynomials in $\poly{\K}{\X}$.
  All the polynomials in $\CancelPoly{p}{q}$ are obtained by multiplying a monomial
  with
  their \kl{S-polynomial} $\spoly{p}{q}$.
\end{lemma}
\begin{proof}
  Let $p,q \in \poly{\K}{\X}$, and let $r \in \CancelPoly{p}{q}$.
  By definition, there exists $\alpha,\beta \in \K$ and $\monelt[n], \monelt[m]
  \in \mon{\X}$ such that $r = \alpha \monelt[n] p + \beta \monelt[m] q$ and
  $\lm(r) < \max(\monelt[n] \lm(p), \monelt[m] \lm(q))$.
  In particular,
  we conclude that $\lm(\monelt[n] p) = \lm(\monelt[m] q)$, and that 
  $\alpha \lc(\monelt[n] p) + \beta \lc(\monelt[m] q) = 0$.

  Let us write $\Delta = \lcm(\lm(p), \lm(q))$.
  Because $\lm(\monelt[n] p) = \lm(\monelt[m] q)$, there exists a monomial 
  $\monelt[l] \in \mon{\X}$ such that 
  $\lm(\monelt[n] p) = \monelt[l] \Delta = \lm(\monelt[m] q)$.
  Furthermore,
  we know that $\lc(p) \beta = - \lc(q) \alpha$.
  As a consequence, one can rewrite $r$ as follows:
  \begin{equation*}
    r = 
    \monelt[l] \alpha \lc(p) 
    \left[
      \frac{\Delta}{\lt(p)} \times p
      - \frac{\Delta}{\lt(q)} \times q
    \right]
    = 
    \monelt[l] \alpha \lc(p) \times \spoly{p}{q} \ .
  \end{equation*}
  We have concluded.
\end{proof}

Remark that the \kl{S-polynomial} is equivariant: if $\gelem \in \group$, then
$\spoly{\gelem \cdot p}{\gelem \cdot q} = \gelem \cdot \spoly{p}{q}$. Given a
set $H$, we write $\intro*\spolyset(H) \defined \bigcup_{p,q \in H}
\rem{H}{\spoly{p}{q}}$. We are now ready to define the saturation algorithm
that will compute \kl{weak equivariant Gröbner bases}, described in
\cref{alg:weakgb}.
Let us remark that \cref{alg:weakgb}
is
an actual algorithm (\cref{lem:weakgb-computable}) that is
\kl(func){equivariant}.

\begin{algorithm}
  \caption{Computing \kl{weak equivariant Gröbner bases} using the algorithm \intro{weakgb}.}
    \label{alg:weakgb}
    \KwIn{An orbit finite set $H$ of polynomials}
    \KwOut{An orbit finite set $\Basis$ that is a \kl{weak equivariant Gröbner basis} of
      $\EqIdlGen{H}$}
    \Begin{
        $\Basis \gets H$\;
        \Repeat{$\Basis$ stabilizes}{
            $\Basis \gets \Basis \cup \spolyset(\Basis)$\;
        }
        \Return{$\Basis$}\;
    }
\end{algorithm}
\begin{lemma}
  \label{lem:weakgb-computable}
  \cref{alg:weakgb} is computable and \kl(func){equivariant},
  and produces an \kl{orbit finite} set $\Basis$ if it terminates.
\end{lemma}
\begin{proof}
  Observe that it is enough to show that $\spolyset{\Basis}$ is orbit-finite for every orbit-finite set $\Basis$.
  First, we compute $\Basis^2$, which is an \kl{orbit finite} set of pairs,
  because $\Basis$ is \kl{orbit finite} and $\Indets$ is 
  \kl{effectively oligomorphic}.
  Then, noting that $\spoly{-\ }{-}$ is computable and \kl(func){equivariant},
  we conclude
  \[
  \bigcup_{p,q \in H}{\spoly{p}{q}}
  \]
  is computable and orbit-finite.
  Now using \Cref{lem:normalisation} one can compute the set $\spolyset(\Basis)$ which is also orbit-finite.
  Furthermore, one can decide whether the set $\Basis$ stabilizes,
  because the membership of a polynomial $p$ in $\Basis$ is decidable,
  since $\group\actson\Indets$ is \kl{effectively oligomorphic} and $\Basis$ is \kl{orbit finite}.
\end{proof}

Let us now use the semantic assumptions to prove the termination of
\cref{alg:weakgb}
(\cref{lem:weakgb-termination})
and the correctness of the resulting orbit finite set
(\cref{lem:weakgb-correctness}).

\begin{lemma}
  \label{lem:weakgb-termination}
  Assume that $(\mon{\X}, \gdivleq)$ is a \kl{WQO}. Then, 
  \cref{alg:weakgb} terminates
  on every \kl{orbit finite} set $H$ of polynomials.
\end{lemma}
\begin{proof}
  Let $\seqof{H_n}[n \in \N]$ be the sequence of (orbit finite) sets of polynomials
  computed by \cref{alg:weakgb}. 
  We associate to each set $H_n$ the set $L_n$ of \kl{characteristic monomials} of the
  polynomials in $H_n$. Because the set of monomials is a \kl{WQO}, and because 
  the sequences are non-decreasing for inclusion, there exists an 
  $n \in \N$ such that, for every $\monelt \in L_{n+1}$, there exists
  $\monelt[n] \in L_n$, such that $\monelt[n] \gdivleq \monelt$.

  We will prove that $H_{n+1} = H_n$ by contradiction. Assume towards this
  contradiction that there exists some $r \in H_{n+1} \setminus H_n$. By
  definition of $H_{n+1}$, there exists $p,q \in H_n$ such that $r \in
  \rem{H_n}{\spoly{p}{q}}$. In particular, $r$ is \kl{normalised} with respect
  to $H_n$. However, because $r \in H_{n+1}$, $\cm(r) \in L_{n+1}$, and
  therefore there exists $\monelt[n] \in L_n$ such that $\monelt[n] \gdivleq
  \cm(r)$. This provides us with a polynomial $t \in H_n$ such that $\cm(t)$
  and an element $\gelem \in \group$ such that $\cm(t) \divleq \gelem \cdot
  \cm(r)$. Because $H_n$ is \kl{equivariant}, we can assume that $\gelem$ is
  the identity. In particular, one concludes that there exists $\monelt[n] \in
  \mon{\X}$ such that $\lm(t) \times \monelt[n] = \lm(r)$, and therefore
  $\lm(t) \gdivleq \cm(r)$. Therefore, one can find some $\alpha \in \K$ such
  that the polynomial $r' \defined r - \alpha \monelt[n] t$ satisfies $r'
  \pmonlt r$, and in particular, $r \toeucl{H_n} r'$.
  This contradicts the fact that $r$ is \kl{normalised} with respect to $H_n$.
\end{proof}

\begin{lemma}
  \label{lem:weakgb-correctness}
  Assume that $\Basis$ is the output of \cref{alg:weakgb}. Then, it 
  is a \kl{weak equivariant Gröbner basis} of the ideal
  $\EqIdlGen{H}$.
\end{lemma}
\begin{proof}
  It is clear that $\Basis$ is a generating set of $\EqIdlGen{H}$, because
  one only add polynomials that are in the ideal generated by $H$ at every step.

  Let $p \in \EqIdlGen{H}$ be a polynomial,
  and let $\mathfrak{d}$ be a \kl{decomposition} of $p$ with respect to
  $\Basis$, that is, a \kl{decomposition} of the form
  \begin{equation}
    p = \sum_{i \in I} \alpha_i \monelt_i p_i
    \quad .
  \end{equation}
  Where $\alpha_i \in \K$, $p_i \in \Basis$, and $\monelt_i \in \mon{\X}$,
  for all $i \in I$.

  Leveraging \cref{lem:chm}, we know that the ordering $\pmonleq$ is
  well-founded. As a consequence, we can consider a minimal 
  decomposition $\mathfrak{d}'$ of $p$ with respect to $\Basis$ such that $\mathfrak{d}'
  \pmonleq \mathfrak{d}$. We now distinguish two cases, depending on whether
  the leading monomial $\lmdec(\mathfrak{d}')$ of the decomposition $\mathfrak{d}'$ is
  equal to the leading monomial of $p$ or not.

  \begin{description}
    \item[Case 1:] $\lmdec(\mathfrak{d}') = \lm(p)$.
      In this case, we conclude immediately,
      as we also have by assumption $\dom(\mathfrak{d}') \subseteq \dom(\mathfrak{d})$.

    \item[Case 2:] $\lmdec(\mathfrak{d}') \neq \lm(p)$.
      In this case, it must be that the set $J$ the set of indices such that
      $\lm(\monelt_i p_i) = \lmdec(\mathfrak{d}')$.
      Let us remark that 
      the sum of leading coefficients 
      of the polynomials in $J$ must vanish: $\sum_{i \in J} \alpha_i \lc(p_i) = 0$.
      As a consequence, the set $J$ has size at least $2$.
      Let us distinguish one element $\star \in J$, and 
      write $J_\star = J \setminus \set{\star}$.
      We conclude that 
      $\alpha_\star = - \sum_{i \in J_\star} \alpha_i \lc(p_i) / \lc(p_\star)$.
      Let us now rewrite $p$ as follows:
      \begin{equation}
        p = \sum_{i \in J_\star} \alpha_i 
        \left(\monelt_i p_i - \frac{\lc(p_i)}{\lc(p_\star)} \monelt_\star p_\star\right)
        + \sum_{i \in I \setminus J} \alpha_i \monelt_i p_i
        \quad .
      \end{equation}
      Now, by definition,
      polynomials $\alpha_i \monelt_i p_i$ for $i \in J \setminus J$ have 
      leading monomials
      strictly smaller than $\monelt[l]$.
      Furthermore,
      the polynomials
      $\monelt_i p_i - \frac{\lc(p_i)}{\lc(p_\star)} \monelt_\star p_\star$ for $i \in J_\star$
      cancel their leading monomials, hence they belong
      to the set $\CancelPoly{p_i}{p_\star}$.
      By \cref{lem:spoly}, we know that these polynomials are obtained by
      multiplying the \kl{S-polynomial} $\spoly{p_i}{p_\star}$ by some monomial.
      Because \cref{alg:weakgb} terminated, we know that 
      $\spoly{p_i}{p_\star} \toeucl{\Basis}^* 0$ by construction.

      By definition of the rewriting relation, we conclude that one can rewrite
      $\spoly{p_i}{p_\star}$ as combination of polynomials in $\Basis$ that
      have smaller or equal leading monomials, and do not introduce new
      indeterminates.

      We conclude that
      the whole sum is composed of polynomials with leading monomials 
      strictly smaller than $\monelt[l]$, and using a subset of the indeterminates
      used in $\mathfrak{d}'$, leading to a contradiction
      because of the minimality of the latter. 
  \end{description}
\end{proof}

As a consequence of the above lemmas, we can now conclude that the 
\cref{alg:weakgb} computes a \kl{weak equivariant Gröbner basis} of the
ideal $\EqIdlGen{H}$, as stated in \cref{thm:weakgb-comput}.

\begin{theorem}
  \label{thm:weakgb-comput}
  Assume that $(\mon[\omega]{\X}, \gdivleq)$ is a \kl{WQO}, and that the order
  $\varleq$ is effectively computable, and that the action of $\group$ is
  \kl{effectively oligomorphic}. 
  Then, the algorithm $\mathsf{weakgb}$ that takes as input an \kl{orbit finite} set $H$ of generators of an
  \kl{equivariant ideal} $\idl$ and computes a \kl{weak equivariant Gröbner
  basis} $\Basis$ of $\idl$.
\end{theorem}
\section{Computing the Equivariant Gröbner Basis}
\label{sec:equivariant-grobner-basis}

\AP The goal of this section is to prove
\cref{thm:compute-egb},
that is, to show that one can effectively compute an \kl{equivariant Gröbner
basis} of an \kl{equivariant ideal}. To that end, we will apply the algorithm
\kl{weakgb} on a slightly modified set of polynomials, and then show that the
result is indeed an \kl{equivariant Gröbner basis}.

\AP
Let us fix a set $\Indets$ of indeterminates equipped with a total ordering
$\varleq$. We define $\IndetsCol \defined \Indets \ordplus \Indets$, that is, the
disjoint union of two copies of $\Indets$, ordered. It will be useful to refer
to the first copy (lower copy) and the second copy (upper copy), noting the
isomorphism between $\IndetsCol$ and $\set{\mathsf{first}, \mathsf{second}}
\times \Indets$, ordered lexicographically, where $\mathsf{first} <
\mathsf{second}$. We will also define $\intro*\forgetCol \colon \IndetsCol \to
\Indets$ that maps a colored variable to its underlying variable.
Beware that $\forgetCol$ is not an order preserving map.
We extend $\forgetCol$ as a morphism from polynomials in
$\poly{\K}{\IndetsCol}$ to polynomials in $\poly{\K}{\Indets}$.

\AP
Given a subset $V \subfin \Indets$, we build the injection
$\intro*\colorWith{V} \colon \Indets \to \IndetsCol$ that maps variables $x$ in
$V$ to $(\mathsf{first}, x)$, and variables $x$ not in $V$ to
$(\mathsf{second}, x)$. Again, we extend these maps as morphisms from
$\poly{\K}{\Indets}$ to $\poly{\K}{\IndetsCol}$. We say that a polynomial $p
\in \poly{\K}{\IndetsCol}$ is \intro{$V$-compatible} if $p \in
\colorWith{V}(\poly{\K}{\Indets})$. Using these definitions, we create
$\intro*\freeColor$ that maps a set $H$ of polynomials to the union over all
finite subsets $V$ of $\Indets$ of the set $\colorWith{V}(H)$. Beware that
$\freeColor$ does not equal $\forgetCol^{-1}$, since we only consider
\kl{$V$-compatible} polynomials (for some finite set $V$).

\AP
We are now ready to write our algorithm to compute an \kl{equivariant Gröbner
basis} by computing the ``congugacy'' $\intro{egb} \defined \forgetCol \circ
\mathop{\kl{weakgb}} \circ \freeColor$. To prove the correctness of our
algorithm, let us first argue that one can indeed compute the \kl{weak
equivariant Gröbner basis} algorithm.

\begin{lemma}
  \label{lem:colored-hypothesis-sat}
  Assume that $\group \actson \Indets$ is
  is \kl{effectively oligomorphic},
  and that $(\mon[\N \times \N]{\Indets}, \gdivleq)$
  is a \kl{well-quasi-order}.
  Then $\kl{egb}$ is a computable function,
  and the function $\kl{weakgb}$ is called 
  on correct inputs.
\end{lemma}
\begin{proof}
  We need to prove that the set $\freeColor(H)$ is computable and 
  \kl{orbit finite}, that $\poly{\K}{\IndetsCol}$ satistfies 
  the \kl{computability assumptions} of \kl{weakgb},
  and that $(\mon{\IndetsCol}, \gdivleq)$ is a
  \kl{well-quasi-ordered} set.
  Finally, we also need to prove that if $H$ is \kl{orbit finite},
  $\forgetCol(H)$ is computable and \kl{orbit finite}. 

  Let us start by proving that $\freeColor(H)$ is computable and \kl{orbit
  finite}. Because $H$ is \kl{orbit finite}, there exists a finite set $H_0
  \subseteq H$ of polynomials such that $\orbit{H_0} = \orbit{H}$. Then, let us
  remark that $\freeColor(H_0)$ can be obtained by considering all finite
  subsets $V$ of variables that appear in $H_0$, which is a computable finite
  set. As a consequence, $\freeColor(H_0)$ is computable, and since
  $\freeColor$ is \kl(func){equivariant}, $\orbit{\freeColor(H_0)} =
  \freeColor(\orbit{H_0}) = \freeColor(H)$.

  Let us now focus on the set $\poly{\K}{\IndetsCol}$. First, it is clear that
  $\group$ is \kl(ord){compatible} with the ordering on $\IndetsCol$ by
  definition of the action, and because $\group$ was compatible with the
  ordering on $\Indets$. Then, $\poly{\K}{\IndetsCol}$ is an \kl{effectively
  oligomorphic} because a polynomial $p$ can be represented in $(\K \times
  (\Indets)^{\leq d} \times (\Indets)^{\leq d})^n$, where $n$ is the number of
  monomials in $p$, and $d$ is the maximal degree of a monomial in $p$.

  Let us now prove that $(\mon{\IndetsCol}, \gdivleq)$ is a
  \kl{well-quasi-ordered} set. To that end, consider a sequence
  $\seqof{\monelt_i}{i \in \N}$ of monomials in $\mon{\IndetsCol}$. A monomial
  in $\mon{\IndetsCol}$ naturally corresponds to a monomial in $\mon[\N \times
  \N]{\Indets}$, where the two exponents are respectively the one of the lower
  copy and the one of the upper copy of the variable.
  Because $(\mon[\N \times \N]{\Indets}, \gdivleq)$ is a
  \kl{well-quasi-ordered} set, we immediately conclude that $(\mon{\IndetsCol}, \gdivleq)$ is a
  \kl{well-quasi-ordered} set.

  Finally, let us prove that $\forgetCol(H)$ is computable and \kl{orbit
  finite}. This is clear because $\forgetCol$ simply consists in forgetting
  the color of the variables.
\end{proof}

Let us now argue that the result of \kl{egb} is indeed a generating set of the
ideal (\cref{lem:correct-gen-set}), and then refine our analysis to
prove that it is an \kl{equivariant Gröbner basis}
(\cref{lem:strong-gb-correct}).

\begin{lemma}
  \label{lem:correct-gen-set}
  Let $H \subseteq \poly{\K}{\Indets}$,
  then $\mathsf{egb}(H)$
  \kl(idl){generates}
  $\EqIdlGen{H}$.
\end{lemma}
\begin{proof}
  Let us remark that $\forgetCol(\freeColor(H)) = H$.
  Since $\kl{weakgb}(\freeColor(H))$
  generates the same ideal as $\freeColor(H)$,
  and since $\forgetCol$ is a morphism,
  we conclude that 
  the set of polynomial
  $\forgetCol(\kl{weakgb}(\freeColor(H)))$
  generates the same ideal as
  $\forgetCol(\freeColor(H)) = H$.
\end{proof}

\begin{lemma}
  \label{lem:strong-gb-correct}
  Let $H \subseteq \poly{\K}{\Indets}$,
  then $\mathsf{egb}(H)$
  is an \kl{equivariant Gröbner basis}
  of $\EqIdlGen{H}$.
\end{lemma}
\begin{proof}
  Let $p \in \IdlGen{H}$,
  $H_\star = \freeColor(H)$,
  $V \defined \dom(p)$,
  $H_V \defined \colorWith{V}(H)$.
  We let $\Basis_\star = \kl{weakgb}(H_\star)$.
  Finally, $\Basis = \forgetCol(\Basis_\star)$.
  It is clear that $\colorWith{V}(p)$
  belongs to $\IdlGen{H_V}$.
  Let us write 
  \begin{equation*}
    \colorWith{V}(p) = \sum_{i=1}^n a_i \monelt_i h_i
  \end{equation*}
  Where $a_i \in \K$, $\monelt_i \in \mon{\IndetsCol}$,
  and $h_i \in \Basis_\star$ is \kl{$V$-compatible}.
  Such a decomposition $\mathfrak{d}$ exists
  because $H_V \subseteq H_\star \subseteq \Basis_\star$.

  Now, because $\Basis_\star$ is a \kl{weak equivariant Gröbner basis} of $\IdlGen{H_\star}$,
  there exists a decomposition $\mathfrak{d}'$ of $\colorWith{V}(p)$
  such that
  $\lm(\colorWith{V}(p)) = \lmdec(\mathfrak{d}') \revlexleq \lmdec(\mathfrak{d})$,
  and 
  $\domdec(\mathfrak{d}') \subseteq \domdec(\mathfrak{d})$.
  In particular, $\mathfrak{d}'$ is a decomposition of $\colorWith{V}(p)$
  using only \kl{$V$-compatible} polynomials in $\Basis_\star$.

  This proves that there exists $\monelt \in \mon{\IndetsCol}$ and $h \in
  \Basis_\star$ such that $\lm(\colorWith{V}(p)) = \lm(\monelt h)$ and $\monelt
  h$ is \kl{$V$-compatible}. In particular, $\monelt h$ must have all its
  variables in $V$, and therefore $\lm(\colorWith{V}(p)) =
  \colorWith{V}(\lm(p)) = \lm(\forgetCol(\monelt h))$. We conclude that
  $\lm(\forgetCol(h)) \divleq \lm(p)$, and that $\forgetCol(h)$ has all its
  variables in $V$.

  We have proven that $\forgetCol(\Basis_\star)$ is 
  an \kl{equivariant Gröbner basis} of $\EqIdlGen{H}$.
\end{proof}

As a consequence, \kl{egb} is the algorithm of
\cref{thm:compute-egb},
and in particular obtain as a corollary that one can decide the \kl{equivariant
ideal membership problem} under our \kl{computability assumptions}, if the set
of indeterminates satisfies that $(\mon[\N \times \N]{\Indets}, \gdivleq)$ is a
\kl{well-quasi-ordered} set. We can leverage these decidability results to
obtain effective representations of \kl{equivariant ideals}, which can then be
used in algorithms as we will see in \cref{sec:applications}.

\begin{corollary}
  \label{cor:equivariant-ideals-computations}
  Assume that $\group \actson \Indets$
  is \kl{effectively oligomorphic},
  and that $(\mon[Y]{\Indets}, \gdivleq)$
  is a \kl{well-quasi-ordered} set
  for every \kl{well-quasi-ordered} set $(Y,\leq)$.
  Then one has an \emph{effective representation} of
  the \kl{equivariant ideals} of $\poly{\K}{\Indets}$,
  such that:
  \begin{enumerate}
    \item One can obtain a representation from a finite set of generators,
    \item One can effectively decide the \kl{equivariant ideal membership problem}
      given a representation,
    \item The following operations are computable at the level of representations:
      the union of two \kl{equivariant ideals}, 
      the product of two \kl{equivariant ideals},
      the intersection of two \kl{equivariant ideals},
      and checking whether two \kl{equivariant ideals} are equal.
  \end{enumerate}
\end{corollary}
\begin{proof}
  Most of this statement follows from \cref{thm:compute-egb}, using
  \kl{equivariant Gröbner bases} as a representation of \kl{equivariant ideals}.
  Indeed, because $\N \times \N$ is a \kl{well-quasi-ordered} set,
  we conclude $(\mon[\N \times \N]{\Indets}, \gdivleq)$ is a 
  \kl{well-quasi-ordered} set too.
  The only non-trivial part is the fact that one can compute an
  \kl{equivariant Gröbner basis} of the
  \emph{intersection} of two \kl{equivariant ideals}.
  To that end, we will adapt the classical argument using 
  \kl{Gröbner bases} to the case of \kl{equivariant Gröbner bases}
  \cite[Chapter 4, Theorem 11]{CLO15}.

  Let $I$ and $J$ be two \kl{equivariant ideals} of $\poly{\K}{\Indets}$,
  respectively represented by \kl{equivariant Gröbner bases} $\Basis_I$ and
  $\Basis_J$. Let $t$ be a fresh indeterminate, and let us consider $\IndetsCol
  \defined \Indets \ordplus \set{t}$, that is, the disjoint union of $\Indets$
  and $\set{t}$, where $t$ is greater than all the variables in $\Indets$.
  
  We construct the \kl{equivariant ideal} $T$ of $\poly{\K}{\IndetsCol}$,
  generated by all the polynomials $t \times h_i$, and $(1-t) \times h_j$,
  where $h_i$ ranges over $\Basis_I$ and $h_j$ ranges over $\Basis_J$. It is
  clear that $T \cap \poly{\K}{\Indets} = I \cap J$.
  Now, because of the hypotheses on $\Indets$, we know that 
  one can compute the \kl{equivariant Gröbner basis} $\Basis_T$ of $T$
  by applying \kl{egb} to the generating set of $T$.
  Finally, we can obtain the \kl{equivariant Gröbner basis} of $I \cap J$ by
  considering $\Basis_T \cap \poly{\K}{\Indets}$, that is, 
  selecting the polynomials of $\Basis_T$ that do not contain the
  indeterminate $t$, which is possible because $\Basis_T$ is an 
  \kl{orbit finite set}
  and $\poly{\K}{\IndetsCol}$ is \kl{effectively oligomorphic}.
\end{proof}
\section{Closure properties}
\label{sec:closure-properties}

In this section, we are interested in listing the operations on sets of
indeterminates equipped with a group action that preserve our \kl{computability
assumptions} and the \kl{well-quasi-ordering} property ensuring that our
\cref{thm:compute-egb} can be applied.  Indeed, it is often tedious to prove
that a given group action $\grp[G] \actson \X$ satisfies the \kl{computability
assumptions} and the \kl{well-quasi-ordering} property, and we aim to provide a
list of operations that preserve these properties, so that simpler examples
(\cref{ex:dlo,ex:eq atoms,ex:dense tree}) can serve as building blocks to model
complex systems.

For the remainder of this section, we fix a pair of group actions $\grp[H]
\actson \X$ and $\grp[G] \actson \Y$, where $\X$ is equipped with a total order
$<_{\X}$ and $\Y$ is equipped with a total order $<_{\Y}$. The constructions
that what we mention in this section were already studied in \cite{GHOLAS24},
via \cite[Example 10]{GHOLAS24} or \cite[Lemma 9]{GHOLAS24}, but they did not
take into account our stronger requirement that $\gdivleq$ should be a
\kl{well-quasi-ordering} for any \kl{well-quasi-ordered} set $(Y, \leq)$ of
exponents on the monomials. We summarise in \cref{tab:closure-properties} the
closure properties of the \kl{computability assumptions} and the
\kl{well-quasi-ordering} property, and we give their definitions of the
operations in \cref{def:union,def:product,def:nested product}.

\begin{table}
  \label{tab:closure-properties}
\centering
\caption{Closure properties of the \kl{computability assumptions} and \kl{well-quasi-ordering} property
for group actions on sets of indeterminates, recapitulating \cref{def:union,def:product,def:nested product}
and \cref{lem:closure-properties-wqo,lem:closure-properties-comp}.}
\begin{tabular}{l|l|l|l|l}
  \toprule
  \textbf{Name} &
  \textbf{Group action} & \textbf{Effective} & \textbf{WQO} & \textbf{Reference} \\
  \midrule
  Sum & $\grp[G] \times \grp[H] \actson \X \ordplus \Y$ &  Yes & Yes & \cref{def:union} \\
  Product &  $\grp[G] \times \grp[H] \actson \X \ordtimes \Y$ & Yes & No & \cref{def:product} \\
  Lex. Product & $\grp[G] \lexGroupAction \grp[H] \actson \X \ordtimes \Y$ & Yes & Yes & \cref{def:nested product} \\
  \bottomrule
\end{tabular}
\end{table}

\begin{definition}[Sum action]\label{def:union}
  Given $(\gelem[\pi], \gelem[\sigma]) \in \grp[G] \times \grp[H]$ 
  and $z \in \X \ordplus \Y$, we define the action
  $(\gelem[\pi], \gelem[\sigma])(z)$ as
  $\gelem[\pi](z)$ if $z \in \X$ and $\gelem[\sigma](z)$ if $z \in \Y$.
  This action is called the \intro{sum action} of $\grp[G]$ and $\grp[H]$ on $\X \ordplus \Y$.
\end{definition}

\begin{definition}[Product action]\label{def:product}
  Given $(\gelem[\pi], \gelem[\sigma]) \in \grp[G] \times \grp[H]$ and $(x,y) \in \X \ordtimes \Y$, we define the \intro{product action} of $\grp[G]$ and $\grp[H]$ on $\X \ordtimes \Y$ as
  $(\gelem[\pi], \gelem[\sigma])\cdot (x,y)$ as $(\gelem[\pi](x), \gelem[\sigma](y))$.
  Note that the ordering of $\X \ordtimes \Y$ is preserved by this action.
\end{definition}

\begin{definition}[Lex. Product Action]\label{def:nested product}
  Let $\grp[G] \intro*\lexGroupAction \grp[H]$ be the group whose elements are of the form 
  $(\gelem,\seqof{\gelem[\sigma]^{x}}[x\in\X])$, where 
  $\gelem[\sigma]_{x} \in \grp[H]$ for every $x\in\X$, and where the multiplication
  is defined as 
  $(\gelem_1,(\gelem[\sigma]_1^{x})_{x\in\X})(\gelem_2,(\gelem[\sigma]_2^{x})_{x\in\X})
  = (\gelem_1\gelem_2, (\gelem[\sigma]_1^{\gelem_2(x)}\gelem[\sigma]_2^x)_{x\in\X})$.
  The \intro{lexicographic product action} of $\grp[G]$ and $\grp[H]$ on $\X \ordtimes \Y$ is defined as
  $\grp[G] \lexGroupAction \grp[H]$ on $\X \ordtimes \Y$ is defined as
  $(\gelem,(\gelem[\sigma]^{x})_{x\in\X})\cdot(x',y') = 
  (\gelem\cdot x', \gelem[\sigma]^{x'}\cdot y')$ for every $(x',y')\in\X \ordtimes \Y$.
Essentially,
each element $x\in\X$ carries its own copy $\{x\}\times\Y$ of the structure $\Y$,
and different copies of the structure $\Y$ can be permuted independently.
\end{definition}

\begin{lemma}\label{lem:closure-properties-wqo}
If $(\mon[Q]{\X}, \gdivleq[\grp[G]])$ and $(\mon[Q]{\Y}, \gdivleq[\grp[H]])$ are \kl{WQOs},
then $(\mon[Q]{\X \ordplus \Y}, \gdivleq[\grp[G] \times \grp[H]])$ and $(\mon[Q]{\X\uplus\Y}, \gdivleq[\grp[G]\otimes\grp[H]])$ are also \kl{WQOs}.
\end{lemma}

\begin{proof}
The \kl{divisibility up to $\grp[G] \times \grp[H]$} order is essentially the product of the orders $\gdivleq[\grp[G]]$ and $\gdivleq[\grp[H]]$,
hence is a \kl{WQO} if both orders are \kl{WQOs} \cite[Lemma 1.5]{SCSC17}.
For, $(\mon[Q]{\X\uplus\Y}, \gdivleq[\G\otimes\calH])$ it follows from \cite[Lemma 9]{GHOLAS24}.
\end{proof}

\begin{lemma}\label{lem:product-not-wqo}
If $(\mon[Q]{\X \ordtimes \Y}, \gdivleq[\grp[G] \times \grp[H]])$ is not a \kl{WQO} even if both $\X$ and $\Y$ are infinite.
\end{lemma}
\begin{proof}
We restate the antichain given in \cite[Example 10]{GHOLAS24},
that will also be used is \cref{rem:multiple-orbits} of \cref{sec:undecidability}
when discussing the undecidability of the \kl{equivariant ideal membership problem}.
Let $\{x_1,x_2,\dots\}$ and $\{y_1,y_2,\dots\}$ be infinite subsets of $\X$ and $\Y$ respectively.
For $n = 3,4,\dots$, let $\monelt{c}_n$ be the monomial
$
\monelt[c]_n = (x_1,y_1)(x_1,y_2)(x_2,y_2)(x_2,y_3)\cdots(x_n,y_n)(x_n,y_1)$.
Then $\setof{\monelt[c]_n}{n = 3,4,\dots}$ is an infinite antichain.
\end{proof}

\begin{lemma}\label{lem:closure-properties-comp}
If $\grp[G]\actson\X$ and $\grp[H]\actson\Y$ satisfy our \kl{computability assumptions},
then so do $(\grp[G] \times \grp[H])\actson (\X \ordplus \Y)$ and $(\grp[G] \otimes \grp[H])\actson (\X \times \Y)$.
\end{lemma}
\begin{proof}
  It is an easy check that the actions defined are all \kl(ord){compatible}
  with the total ordering on the set of indeterminates.
  Let us now sketch the proof that points (2) and (3) of our
  \kl{computability assumptions} hold for the action $(\grp[G] \times \grp[H])\actson (\X \ordplus \Y)$,
  the proof being similar for $(\grp[G] \otimes \grp[H])\actson (\X \times \Y)$.
  We represent elements of $\X \ordplus \Y$ as elements of a tagged union,
  and given a word $w \in (\X \ordplus \Y)^*$, we can 
  create three words: one $w_\X$ in $\X^*$, one $w_\Y$ in $\Y^*$, and one $w_\text{tag}$ in $\set{ 0, 1 }^*$,
  obtained by considering respectively only the elements of $\X$,
  only the elements of $\Y$, and restricting the elements of $\X \ordplus \Y$ to whether 
  or not they belong to $\X$.
  It is an easy check that two words $u,v$ are in the same orbit if and only if
  their $u_\X$ and $v_\X$ are in the same orbit under $\grp[G]$,
  their $u_\Y$ and $v_\Y$ are in the same orbit under $\grp[H]$, and their $u_\text{tag}$ and $v_\text{tag}$ are equal.

  To list representatives of the orbits in $(\X \ordplus \Y)^n$ for a fixed 
  $n \in \N$, we can list representatives $u_\X$ of the orbits in $\X^{\leq n}$,
  representatives $u_\Y$  of the orbits in $\Y^{\leq n}$,
  and words $u_\text{tag} \in \set{0,1}^n$, and consider triples
  $(u_\X, u_\Y, u_\text{tag})$ such that $|u_\X| + |u_\Y| = n$,
  $|u_\text{tag}|_0 = |u_\X|$, and $|u_\text{tag}|_1 = |u_\Y|$.
\end{proof}
\section{Applications}
\label{sec:applications}

\paragraph{Polynomial computations.} \AP The fact that (finite control) systems
performing polynomial computations can be verified follows from the theory of
\kl{Gröbner bases} on finitely many indeterminates \cite{MULSEI02,BEDUSHWO17}.
There were also numerous applications to automata theory, such as deciding
whether a weighted automaton could be determinised (resp. desambiguated)
\cite{BESM23,PUSM24}. We refer the readers to a nice survey recapitulating the
successes of the ``Hilbert method'' automata theory \cite{BOJAN19}. A natural
consequence of the effective computations of \kl{equivariant Gröbner bases} is
that one can apply the same decision techniques to \emph{orbit finite
polynomial computations}. For simplicity and clarity, we will focus on
\kl{polynomial automata} without states or zero-tests \cite{BEDUSHWO17}, but
the same reasoning would apply to more general systems as we will discuss in
\cref{rem:topological-wsts}.

\AP Before discussing the case of orbit finite polynomial automata, let us
recall the setting of \kl{polynomial automata} in the classical case, as
studied by \cite{BEDUSHWO17}, with techniques that dates back to
\cite{MULSEI02}. A \intro{polynomial automaton} is a tuple $A \defined (Q,
\Sigma, \delta, q_0, F)$, where $Q = \K^n$ for some finite $n \in \N$, $\Sigma$
is a finite alphabet, $\delta \colon Q \times \Sigma \to Q$ is a transition
function such that $\delta(\cdot,a)_i$ is a polynomial in the indeterminates
$q_1, \dots, q_n$ for every $a \in \Sigma$ and every $i \in \set{1, \dots, n}$,
$q_0 \in Q$ is the initial state, and $F \colon Q \to \K$ is a polynomial
function describing the final result of the automaton. The \intro{zeroness
problem for polynomial automata} is the following decision problem: given a
\kl{polynomial automaton} $A$, is it true that for all words $w \in \Sigma^*$,
the polynomial $F(\delta^*(q_0, w))$ is zero? It is known that the \kl{zeroness
problem for polynomial automata} is decidable \cite{BEDUSHWO17}, using the
theory of \kl{Gröbner bases} on finitely many indeterminates.

\newcommand{\toequiv}{\stackrel{\text{\tiny{eq}}}{\to}}

\AP Let us now propose a new model of computation called \kl{orbit finite
polynomial automata}, and prove an analogue decidability result. Let us fix an
\kl{effectively oligomorphic} action $\group \actson \Indets$, such that there
exists finitely many indeterminates $V \subfin \Indets$ such that $\group$ acts
as the identity on $V$. Given such a function $f \colon \Indets \to \K$, and
given a polynomial $p \in \poly{\K}{\Indets}$, we write $p(f)$ for the
evaluation of $p$ on $f$, that belongs to $\K$. Let us emphasis that the model
is purposely designed to be simple and illustrate the usage of \kl{equivariant
Gröbner bases}, and not meant to be a fully-fledged model of computation.

\begin{definition}
  \label{def:orbit-finite-polynomial-automaton}
  An \intro{orbit finite polynomial
  automaton} is a tuple $A \defined (Q, \delta, q_0, F)$, where $Q =
  \Indets \to \K$, $q_0 \in Q$ is a function that is non-zero for finitely
  many indeterminates, $\delta \colon
  \Indets \times \Indets \toequiv \poly{\K}{\Indets}$ 
  is a
  polynomial update function, and $F \in \poly{\K}{V}$ is a polynomial 
  computing the result of the automaton. 

  Given a letter $a \in \Indets$ and a
  state $q \in Q$, the updated state $\delta^*(a,q)  \in Q$ is defined as the function from
  $\Indets$ to $\K$ defined by $\delta^*(a,q) \colon x \mapsto \delta(a,x)( q )$.
  The update function is naturally extended to words. Finally, the
  output of an \kl{orbit finite polynomial automaton} on a word $w \in \Indets^*$
  is defined as $F(\delta^*(w,q_0))$.
\end{definition}

\AP \kl{Orbit finite polynomial automata} can be used to model programs that
read a string $w \in \Indets^*$ from left to right, having as internal state a
dictionary of type \texttt{dict[indet, number]}, which is updated using
polynomial computations. As for \kl{polynomial automata}, the
\intro(ofpa){zeroness problem} for orbit finite polynomial automata is the
following decision problem: decide if for every input word $w$, the output
$F(\delta^*(w, q_0))$ is zero.

\begin{theorem}
  \label{cor:orbit-finite-polynomial-automata-zeroness}
  Let $\Indets$ be a set of indeterminates that satisfies the
  \kl{computability assumptions} and such that $(\mon[Y]{\Indets}, \gdivleq)$ is a
  \kl{well-quasi-ordering}, for every \kl{well-quasi-ordered} set $(Y, \leq)$.
  Then, the \kl(ofpa){zeroness problem} is decidable for \kl{orbit finite polynomial automata}.
\end{theorem}
\begin{proof}
  Let $A = (Q, \delta, q_0, F)$ be an \kl{orbit finite polynomial
  automaton}. Following the classical \emph{backward procedure} for such
  systems, we will compute a sequence of sets $E_0 \defined \setof{ q \in Q }{
  F(q) = 0 }$, and $E_{i+1} \defined \mathrm{pre}^\forall(E_i) \cap E_i$, where
  $\mathrm{pre}^\forall(E)$ is the set of states $q \in Q$ such that for every
  $a \in \Sigma$, $\delta^*(q,a) \in E$. We will prove that the sequence of
  sets $E_i$ stabilises, and that it is computable. As an immediate
  consequence, it suffices to check that $q_0 \in E_{\infty}$, where $E_\infty$
  is the limit of the sequence $(E_i)_{i \in \N}$, to decide the
  \kl(ofpa){zeroness problem}.

  The only idea of the proof is to notice that all the sets $E_i$ are
  representable as zero-sets of \kl{equivariant ideals} in
  $\poly{\K}{\Indets}$, allowing us to leverage the effective computations of
  \cref{cor:equivariant-ideals-computations}. Given a set $H$ of polynomials,
  we write $\mathcal{V}(H)$ the collections of states $q \in Q$ such that $p(q)
  = 0$ for all $p \in H$.
  It is easy to see that $E_0 = \mathcal{V}(\set{F}) = \mathcal{V}(\idl_0)$,
  where $\idl_0$ is the \kl{equivariant ideal} generated by $F$, since 
  $F \in \poly{\K}{V}$ and $V$ is invariant under the action of $\group$.
  Furthermore, assuming that $E_i = \mathcal{V}(\idl_i)$, we can
  see that 
  \begin{align*}
    \mathrm{pre}^\forall(E_i) 
    & = \setof{ q \in Q }{ \forall a \in \Sigma, \delta^*(a,q) \in E_i } \\
    & = \setof{ q \in Q }{ \forall a \in \Sigma, \forall p \in \idl_i, p(\delta^*(a,q)) = 0 } \\
    & = \setof{ q \in Q }{ \forall p' \in \idl[J], p'(q) = 0 }
  \end{align*}
  Where, the \kl{equivariant ideal} $\idl[J]$ is generated by the
  polynomials $\mathrm{pullback}(p,a) \defined p [ x \mapsto \delta(a,x)]$
  for every pair $(p, a) \in \idl_i \times \Sigma$. 
  As a consequence, we have $E_{i+1} = \mathcal{V}(\idl_{i+1})$, where
  $\idl_{i+1} = \idl_i + \idl[J]$.
  Because the sequence $\seqof{ \idl_i }[ i \in \N]$ is increasing, and thanks
  to the \kl{equivariant Hilbert basis property} of $\poly{\K}{\Indets}$, there
  exists an $n_0 \in \N$ such that $\idl_{n_0} = \idl_{n_0 + 1} = \idl_{n_0 +
  2} = \cdots$. In particular, we do have $E_{n_0} = E_{n_0 + 1} = E_{n_0 + 2}
  = \cdots$.

  Let us argue that we can compute the sequence $\idl_i$.
  First,  $\idl_0 = \EqIdlGen{F}$ is finitely represented.
  Now, 
  given an \kl{equivariant ideal} $\idl$, represented by an \kl{orbit finite}
  set of generators $H$,
  we can compute the \kl{equivariant ideal} $\idl[J]$ generated by the
  polynomials $\mathrm{pullback}(p,a) \defined p [ x_i \mapsto \delta(a)(x_i)]$
  for every pair $(p, a) \in H \times \Sigma$. Indeed, $H \times \Sigma$ is
  \kl{orbit finite}, and the function $\mathrm{pullback}$ is
  computable and \kl(func){equivariant}: given $\gelem \in \group$, we can
  show that
  \begin{align*}
    \gelem \cdot \mathrm{pullback}(p, a) & = 
    \gelem \cdot (p [ x_i \mapsto \delta(a,x_i)]) & \text{ by definition }\\
                                                  & = p [ x_i \mapsto (\gelem \cdot \delta(a, x_i))] 
                                                  & \text{ $\gelem$ acts as a morphism } \\
    & = p [ x_i \mapsto \delta(\gelem \cdot a, \gelem \cdot x_i))] 
    & \text{ $\delta$ is \kl(func){equivariant} } \\
    & = (\gelem \cdot p) [ x_i \mapsto \delta(\gelem \cdot a, x_i)] 
    & \text{ definition of substitution }
    \\
    & = \mathrm{pullback}(\gelem \cdot p, \gelem \cdot a).
    & \text{ by definition }
  \end{align*}
  
  Finally, one can detect when the sequence stabilises, by checking whether
  $\idl_i = \idl_{i+1}$, which is decidable because the
  \kl{equivariant ideal membership problem} is decidable 
  by \cref{thm:compute-egb}.

  To conclude, it remains to check whether $q_0 \in E_\infty$,
  which amounts to check that $q_0 \in \mathcal{V}(\idl_\infty)$.
  This is equivalent to checking whether for every element $p \in \Basis$
  where $\Basis$ is an \kl{equivariant Gröbner basis} of $\idl_\infty$, we have
  $p(q_0) = 0$, which can be done by enumerating relevant orbits.
\end{proof}

The \kl{orbit finite polynomial automata} model could be extended to allow for
inputs of the form $\Indets^k$ for some $k \in \N$, or even be recast in the
theory of nominal sets \cite{BOJAN16inf}. Furthermore, leveraging the closure
properties of \cref{lem:closure-properties-comp,lem:closure-properties-wqo},
one can also reduce the equivalence problem for \kl{orbit finite polynomial
automata} to the \kl(ofpa){zeroness problem}, by considering the \kl{sum
action} on the registers to compute the difference of the two results. We leave
a more detailed investigation of the generalisation of \kl{polynomial automata}
to the orbit finite setting for future work.

\begin{remark}
  \label{rem:topological-wsts}
  The proof of \cref{cor:orbit-finite-polynomial-automata-zeroness} can be
  recast in the more general setting of 
  \intro{topological well-structured transition system}, that were introduced by
  Goubault-Larrecq in \cite{JGL07}, who noticed that the pre-existing notion of
  \emph{Noetherian space} could serve as a topological generalisation of
  \kl{Noetherian rings} (where ideal-based method can be applied),
  and 
  \kl{well-quasi-orderings}, for which the celebrated decision procedures on
  \emph{well-structured transition systems} can be applied \cite{ABDU96}. In particular,
  Goubault-Larrecq used such systems to verify properties of \emph{polynomial
  programs} computing over the complex numbers, that can communicate over lossy
  channels using a finite alphabet \cite{JGL10}. 
  Because of \cref{cor:equivariant-ideals-computations}, we do have an 
  effective way to compute on the topological spaces at hand, 
  and therefore we can apply the theory of
  \kl{topological well-structured transition systems} to verify systems
  such as \emph{orbit finite polynomial automata communicating using a finite alphabet
  over lossy channels}.
  We refer to \cite[Chapter 9]{JGL13} for a survey on the theory of 
  Noetherian spaces.
\end{remark}

\paragraph{Reachability problem of symmetric data Petri nets.} The classical
model of Petri nets was extended to account for arbitrary data attached to
tokens to form what is called data Petri nets. We will not discuss the precise
definitions of these models, but point out  that a reversible data Petri net is
exactly what is called a \kl{monomial rewriting system} \cite[Section
8]{GHOLAS24}. Because reachability in such rewriting systems can be decided
using \kl{equivariant ideal membership} queries \cite[Theorem 64]{GHOLAS24}, we
can use \cref{thm:compute-egb} and \cref{lem:reducts-equiv-hilbert} to show
\cref{cor:rev data VAS}. Note that \kl{monomial rewrite systems} will be at the
center of our undecidability results in \cref{sec:undecidability}.

\begin{corollary}\label{cor:rev data VAS}
  For every \kl{nicely orderable} group action $\group\actson\Indets$,
  the reachability problem for reversible Petri nets with data in $\Indets$ is decidable.
\end{corollary}

\paragraph{Orbit-finite systems of equations} The classical theory of solving
finite systems of linear equations has been generalised to the infinite setting
by \cite{GHL22}, \cite[Section 9]{GHOLAS24}. In this setting, one considers an
\kl{effectively oligomorphic} group action $\group\actson\Indets$, and the
vector space $\lin(\Indets^n)$ generated by the indeterminates $\Indets^n$ over
$\K$. An orbit-finite system of equations asks whether a given vector $u \in
\lin(\Indets^n)$ is in the vector space generated by an \kl{orbit-finite} set
of vectors $V$ in $\lin(\Indets^n)$ \cite[Section 9]{GHOLAS24}. It has been
shown that the solvability of these systems of equations reduces to the
\kl{equivariant ideal membership problem} \cite[Theorem 68]{GHOLAS24}, and as a
consequence of this reduction and
\Cref{thm:compute-egb,lem:reducts-equiv-hilbert} we get that:

\begin{corollary}\label{cor:lin solv}
  For every \kl{nicely orderable} group action $\group\actson\Indets$,
  the solvability problem for orbit-finite systems of equations
  is decidable.
\end{corollary}
Note that the above corollary is an extension of \cite[Theorem 6.1]{GHL22} to all nicely orderable group actions.
\section{Undecidability Results}
\label{sec:undecidability}

In this section, we aim to show that the \kl{equivariant ideal membership
problem} is undecidable under the usual \kl{computability assumptions} on the
group action, when we do not assume that $(\mon{\Indets}, \gdivleq)$ is a
\kl{well-quasi-ordering}. In particular, this would show that computing
\kl{equivariant Gröbner bases} is not possible in these settings, proving the
optimality of our decidability \cref{thm:compute-egb}. Beware that there are
some pathological cases where the \kl{equivariant ideal membership problem} is
easily decidable, even when $(\mon{\Indets}, \gdivleq)$ is not a
well-quasi-ordering, as illustrated by the following
\cref{ex:non-wqo-undecidable}, and it is not possible to obtain such a
dichotomy result without further assumptions on the group action.

\begin{example}
  \label{ex:non-wqo-undecidable}
  Let $\Indets = \{x_1, x_2, \ldots\}$ be an infinite set of indeterminates,
  and let $\group$ be trivial group acting on $\Indets$.
  Then, the \kl{equivariant ideal membership problem} is decidable.
  Indeed, since the group is trivial, whenever one provides a finite set
  $H$ of generators of an \kl{equivariant ideal} $I$, one can
  in fact work in $\poly{\K}{V}$, where $V$ is the set of indeterminates
  that appear in $H$.
  Then, the \kl{equivariant ideal membership problem} is reduces to 
  the \kl{ideal membership problem} in $\poly{\K}{V}$, which is decidable.
\end{example}

\AP However, one we are able to prove the undecidability of the \kl{equivariant
ideal membership problem} under the assumption that the set of indeterminates
$\Indets$  contains an \intro(of){infinite path} $P \defined \seqof{x_i}[i \in
\N] \subseteq \Indets$, that is, a set of indeterminates such that $(x_i,x_j)
\in P^2$ is in the same orbit as $(x_0, x_1)$ if and only if $|i - j| = 1$, for
all $i,j \in \N$. We similarly define \reintro(of){finite paths} by considering
finitely many elements. The prototypical example of a set of indeterminates
containing an \kl(of){infinite path} is $\Indets = \Z$ equipped with the group
$\group$ of all shifts. The presence of an \kl(of){infinite path} clearly
prevents $(\mon{\Indets}, \gdivleq)$ from being a \kl{well-quasi-ordering}, as
shown by the following \cref{rem:not-wqo}. Furthermore, for indeterminates
obtained by considering \kl{homogeneous structures} and their automorphism
groups (\cref{sec:act ex}), the presence of an \kl(of){infinite path} has been
conjectured to be a necessary and sufficient condition for $(\mon{\Indets},
\gdivleq)$ to be a \kl{well-quasi-ordering}: this follows from a conjecture of
Schmitz restated in \cref{conj:wqo-infinite-path}, that generalises one of
Pouzet (\cref{rem:conj-wqo-pouzet}), as explained in
\cref{rem:conj-wqo-infinite-path}.

\begin{remark}
  \label{rem:not-wqo}
  Assume that $\Indets$ contains an \kl(of){infinite path}
  $P \defined \seqof{x_i}[i \in \N]$.
  Then, the set of monomials $\setof{x_0^3 x_1^1 \cdots x_{n-1}^1 x_n^2}{n \in \N}$
  is an infinite antichain in $(\mon{\Indets}, \gdivleq)$.
  Indeed, assume that there exists $n < m$, and a group element $\gelem \in \group$ such that
  $\gelem \cdot \monelt_n \divleq \monelt_m$.
  Then, $\gelem \cdot x_0 = x_0$, because it is the only indeterminate with 
  exponent $3$ in $\monelt_m$. Furthermore, 
  $\gelem \cdot (x_0,x_1) = (x_i,x_j)$ implies that 
  $|i - j| = 1$, and since $\gelem \cdot x_0 = x_0$, we conclude
  $\gelem \cdot x_1 = x_1$. By an immediate induction, we 
  conclude that $\gelem \cdot x_i = x_i$ for all $0 \leq i \leq n$,
  but then we also have that the degree of $\gelem \cdot x_n$ is less than $2$
  in $\monelt_m$, which contradicts the fact that $\gelem \cdot \monelt_n \divleq \monelt_m$.
\end{remark}

\begin{conjecture}[Schmitz]
  \label{conj:wqo-infinite-path}
  Let $\mathcal{C}$ be a class of finite relational structures. Then, the following are
  equivalent:
  \begin{enumerate}
    \item The class of structures of $\mathcal{C}$ labelled with 
      any \kl{well-quasi-ordered} set $(Y, \leq)$ is
      itself \kl{well-quasi-ordered} under the
      labelled-induced-substructure relation.
    \item For every existential formula $\varphi(x,y)$,
      there exists $N_\varphi \in \N$, such 
      that $\varphi$ does not \kl(efo){define paths} of length greater than $N_\varphi$
      in the structures of $\mathcal{C}$.
  \end{enumerate}
  Where a formula \intro(efo){defines a path} of length $n$ in a structure
  if there exists $n$ distinct elements $a_0, \ldots, a_{n-1}$ in the structure
  such that $\varphi(a_i, a_j)$ holds if and only if $|i - j| = 1$.
\end{conjecture}

\begin{remark}
  \label{rem:conj-wqo-pouzet}
  The conjecture of Schmitz is a generalization of Pouzet's conjecture
  \cite{POUZ72} that states that a class $\mathcal{C}$  of finite relational structures is
  \kl{well-quasi-ordered} under the labelled induced-substructure relation for
  every \kl{well-quasi-ordered} set of labels, 
  if and only if it is the case for the set of two incomparable labels
  \cite[Problem 9]{POUZ24}. A negative answer to Pouzet's conjecture
  has been obtained in \cite{KRTH90,KS91} for finite (non-relational) structures,
  but the conjecture remains open for finite relational structures.
\end{remark}

\begin{remark}
  \label{rem:conj-wqo-infinite-path}
  Let $\Indets$ be an infinite \kl{homogeneous structure},
  such that $(\mon{\Indets}, \gdivleq)$ is not a \kl{well-quasi-ordering}.
  Then, the collection of finite substructures of $\Indets$
  labelled by $(\N,\leq)$ is not \kl{well-quasi-ordered} under the
  labelled-induced-substructure relation.
  Hence, if one believes that \cref{conj:wqo-infinite-path} holds,
  there exists an existential formula $\varphi(x,y)$ such that
  $\varphi$ defines arbitrarily long paths in $\Indets$.
  Because $\Indets$ is \kl{homogeneous},
  this means that $\varphi$ defines an infinite path in $\Indets$,
  and in particular, 
  $\Indets$ contains an \kl(of){infinite path} $P$, as introduced
  for generic sets of indeterminates.
\end{remark}

As already mentioned in \cref{rem:conj-wqo-infinite-path}, it is conjectured
that the presence of an \kl(of){infinite path} is a necessary condition for the
\kl{equivariant ideal membership problem} to be undecidable in the case of
\kl{homogeneous structures} over relational signatures. Let us briefly argue
that in the case of \kl{homogeneous} \emph{$3$-graphs} $\mathcal{G}_3$ (i.e. a
structure with three distinct edge relations), the \emph{WQO dichotomy theorem}
\cite[Theorem 4]{LAPIO20}, exactly states that: either
$(\mon[Y]{\mathcal{G}_3}, \gdivleq)$ is a \kl{well-quasi-ordering} for all
\kl{well-quasi-ordered} sets $Y$, or there exists an \kl(of){infinite path} in
$\mathcal{G}_3$. We conclude that for \kl{homogeneous} \emph{$3$-graphs},
either the \kl{equivariant ideal membership problem} is undecidable
(\cref{thm:undecidable-paths}), or our \cref{thm:compute-egb} can be applied to
compute \kl{equivariant Gröbner bases}.

\paragraph{Monomial Reachability}
The undecidability results we will present in this section regarding the
\kl{equivariant ideal membership problem} will use the polynomials in a very
limited way: we will only need to consider \emph{monomials}, and there will
even be a bound on the maximal exponent used. Before going into the details of
our reductions, let us first introduce an intermediate problem that will be
easier to work with: the (equivariant) \kl{monomial reachability problem}. 

\begin{definition}
  \label{def:mon-rewrite-system}
  A \intro{monomial rewrite system} is a finite set of pairs of the form
  $\set{\monelt, \monelt'}$ where $\monelt, \monelt' \in \mon{\Indets}$.
  The \intro{monomial reachability problem} is the problem of deciding whether
  there exists a sequence of rewrites that transforms $\monelt_s$ into $\monelt_t$
  using the rules of a monomial rewrite system $R$, where
  a \intro(monrew){rewrite step} is a pair of the form
  \begin{equation*}
    \monelt[n] (\gelem \cdot \monelt)
    \leftrightarrow_R 
    \monelt[n] (\gelem \cdot \monelt')
    \text{ if } \set{\monelt, \monelt'} \in R
    \text{ and } \gelem \in \group
    \quad .
  \end{equation*}
\end{definition}

\begin{example}
  \label{ex:mon-rewrite-system}
  Let $\Indets = \N$ and $\group$ be the set of all bijections of $\Indets$.
  Then, the rewrite system $x_1^2 x_2^2 \leftrightarrow_R x_1^2$
  satisfies $\monelt \leftrightarrow_R^* x_1^2$ if and only if 
  $\monelt$ has all its exponents that are multiple of $2$.
\end{example}

The following \cref{lem:mon-rewrite-red-membership} shows that the \kl{monomial
reachability problem} can be reduced to the \kl{equivariant ideal membership
problem}, and follows the exact same reasoning as in the case of finitely many
indeterminates \cite{MAME82}. This reduction was also noticed in \cite[Theorem
64]{GHOLAS24}.

\begin{lemma}[label=lem:mon-rewrite-red-membership,restate=lem:mon-rewrite-red-membership]
  \AP
  One can solve the \kl{monomial reachability problem}
  provided that one can solve the \kl{equivariant ideal membership problem}.
  \proofref{lem:mon-rewrite-red-membership}
\end{lemma}

In order to show that the \kl{equivariant ideal membership problem} is
undecidable, it is therefore enough to show that the \kl{monomial reachability
problem} is undecidable. To that end, we will encode the Halting problem of a
Turing machine. There are two main obstacles to overcome: first, the
reversibility of the rewriting system, which can be (partially) solved by
considering a \emph{reversible version} of a \emph{deterministic} Turing
machines, as explained in \cite[Simulation by bidirected systems, p.
15]{GAMAPASCZE22}; and second, the fact that the configurations of the Turing
machine cannot staightforwardly be encoded as monomials due to the
commutativity of the multiplication.

\paragraph{Structures Containing Paths.} \AP Let us assume for the rest of this
section that $\Indets$ is a set of indeterminates that contains an
\kl(of){infinite path}, let us fix a binary alphabet $\Sigma \defined
\set{a,b}$. Given a \kl(of){finite path} $P \defined \seqof{x_i}[0 \leq i <
4n]$, we define a function $\intro*\wenc{ \cdot}_P \colon \Sigma^{\leq n} \to
\mon{\Indets}$, where $\Sigma$ is a finite alphabet, that \intro[word encoding]{encodes a word} $u \in
\Sigma^{\leq n}$ as a monomial. Namely, we define inductively
$\wenc{\varepsilon} \defined 1$, $\wenc{a u}_P = x_0^4 x_1^2 x_2^1 x_3^3
(\mathsf{shift}_{+4} \cdot \wenc{u}_P)$ and $\wenc{b u}_P = x_0^4 x_1^1 x_2^2
x_3^3 (\mathsf{shift}_{+4} \cdot \wenc{u}_P)$ for all $u \in \Sigma^*$, where
$\mathsf{shift}_{+k}$ acts on $P$ by shifting the indices by
$k$.\footnote{There may be no element $\gelem \in \group$ that acts like
$\mathsf{shift}_{+1}$, we only use it as a function.} Let us remark that
\kl{monomial rewriting} applied on \kl{word encodings} can simulate
(reversible) string rewriting on words of a given size.

\begin{lemma}
  \label{lem:word-encoding-string-subst}
  Let $P,Q$ be two \kl(of){finite paths} in $\Indets$,
  such that $(p_0,p_1)$ is in the same orbit as 
  $(q_0,q_1)$.
  Let $u,v,w \in \Sigma^*$ be three words, such that $|u| = |v| \leq |w|$,
  and let $\monelt[n] \in \mon{\Indets}$ be a monomial.
  Assume that there exists $\gelem \in \group$
      such that $\wenc{w}_P = \monelt[m] (\gelem \cdot \wenc{u}_Q)$,
       $\monelt[n] = \monelt[m] (\gelem \cdot \wenc{v}_Q)$,
  and that $\wenc{w}_P$, $\wenc{u}_Q$ and $\wenc{v}_Q$
  are well-defined.
  Then,
      there exists $x, y \in \Sigma^*$
      such that $x u y = w$ and $\wenc{x v y}_P = \monelt[n]$.
\end{lemma}
\begin{proof}
  Let us write $\gelem \cdot q_0 = p_k$ for some $k \in \N$.
  Because the only indeterminates with degree $4$ in $\wenc{w}_P$ are
  the ones of the form $p_{4i}$, we have that $k$ is a multiple of $4$
  (i.e. at the start of a letter block).
  Since $(q_0, q_1)$ is in the same orbit as $(p_0, p_1)$,
  and both $P$ and $Q$ are \kl(of){finite paths},
  we conclude that $\gelem \cdot (q_0, q_1) = (p_{4i}, p_{4i+1})$
  or $\gelem \cdot (q_0, q_1) = (p_{4i+1}, p_{4i-1})$.
  Applying the same reasoning, thrice, 
  we have either $\gelem \cdot (q_0, q_1, q_2, q_3) = (p_{4i}, p_{4i+1}, p_{4i+2}, p_{4i+3})$
  or $\gelem \cdot (q_0, q_1, q_2, q_3) = (p_{4i}, p_{4i-1}, p_{4i-2}, p_{4i-3})$.
  However, in the second case, the exponent of $p_{4i-3}$ in $\wenc{w}_P$ is at most $2$,
  which is incompatible with the fact that the one of $q_3$ in $\wenc{u}_Q$ is $3$.
  By induction on the length of $u$, we immediately obtain that 
  $\gelem \cdot \wenc{u}_Q = \mathsf{shift}_{+4i} \cdot \wenc{u}_P$ and
  therefore that 
  $w = x u y$ for some $x,y \in \Sigma^*$.
  Finally, because $\wenc{v}_Q$ uses exactly the same indeterminates as 
  $\wenc{u}_Q$, we can also conclude that
  $\wenc{xvy}_P = \monelt[n]$.
\end{proof}

\Cref{lem:word-encoding-string-subst} shows that all encodings
using \kl(of){finite paths} with the same initial orbit are compatible with
each other for the purpose of \kl{monomial rewriting}. Let us now assume that
the alphabet is any finite set of letters, using a suitable unambiguous
encoding of the alphabet in binary \cite{BERST09}. This bigger alphabet size
will simplify the statement and proof of the following
\cref{lem:reversible-machine}, which explains how to simulate a
reversible Turing machine using \kl{monomial rewriting}. Given a reversible
Turing machine $M$ with a finite set $Q$ of states and tape alphabet $\Sigma$,
we will consider the following alphabet $\Gamma \defined \set{ \triangleleft,
  \triangleright } \times \set{ \text{pre}, \text{run}, \text{post} } \uplus Q
  \uplus \Sigma \uplus \set{ \square, \square_1, \square_2}$. The letter
  $\square$ is a blank symbol, and the letters $\triangleleft$ and
  $\triangleright$ are used to delimit the beginning and the end of the tape,
  with some extra ``phase information''. In a first \kl{monomial rewrite
  system}, we will encode a run of a reversible Turing machine $M$ on a fixed
  size input tape (\cref{lem:reversible-machine}), and in a second
  \kl{monomial rewrite system}, we will create a tape of arbitrary size
  (\cref{lem:tape-creation}). The union of these two \kl{monomial
  rewrite systems} will then be used to prove the undecidability of the
  \kl{equivariant ideal membership problem} in \cref{thm:undecidable-paths}.

\begin{lemma}
  \label{lem:reversible-machine}
  Let us fix $(x_0, x_1)$ a pair of indeterminates.
  There exists a
  \kl{monomial rewrite system} $R_M$ such that the following
  are equivalent for every $n \geq 1$,
  and for any \kl(of){finite path} $P$ of length $4(n+2)$ 
  such that $(p_0, p_1)$ is in the same orbit as $(x_0, x_1)$:
  \begin{enumerate}
    \item $\wenc{ \triangleright^{\text{run}} q_0 \square^{n-1}
                  \triangleleft^{\text{run}}
     }_P \leftrightarrow_{R_M}^* 
     \wenc{ \triangleright^{\text{run}} q_f \square^{n-1}
                  \triangleleft^{\text{run}} }_P$,
      \item $M$ halts on the empty word using a tape bounded by $n-1$ cells.
  \end{enumerate}
  Furthermore, every monomial that is 
  reachable from $\wenc{ \triangleright^{\text{run}} q_0 \square^{n-1} \triangleleft^{\text{pre}} }_P$
  or $\wenc{ \triangleright^{\text{run}} q_f \square^{n-1} \triangleleft^{\text{run}} }_P$
  is the image of a word of the form
  $\wenc{\triangleright^{\text{run}} u \triangleleft^{\text{run}}}_P$  
  where $u \in (Q \uplus \Sigma \uplus \square)^n$.
\end{lemma}
\begin{proof}
  Transitions of the deterministic reversible Turing machine using bounded tape size can be 
  modelled as a reversible string rewriting system using finitely many rules 
  of the form $u \leftrightarrow v$, where $u$ and $v$ are words
  over $(Q \uplus \Sigma \uplus \square)$ having the same length $\ell$
  For each rule $u \leftrightarrow v$, we create rules 
  $\wenc{u}_P \leftrightarrow_{R_M} \wenc{v}_P$ 
  for every \kl(of){finite path} $P$ of length $4l$.
  Note that there are only orbit finitely many such \kl(of){finite paths} $P$,
  and one can effectively list some representatives,
  because $\Indets$ is \kl{effectively oligomorphic}.
  This system is clearly complete, in the sense that one can perform a substitution
  by applying a monomial rewriting rule, but \cref{lem:word-encoding-string-subst}
  also tells us it is correct, in the sense that it cannot perform anything else
  than string substitutions.
  Furthermore, we can assume 
  that the reversible Turing machine
  starts with a clean tape and ends with a clean tape.
\end{proof}

\Cref{lem:reversible-machine} shows that one can simulate the
runs, provided we know in advance the maximal size of the tape used by the
reversible Turing machine. The key ingredient that remains to be explained is
how one can start from a finite monomial $\monelt$ and create a tape of
arbitrary size using a \kl{monomial rewrite system}. The difficulty is that we
will not be able to ensure that we follow one specific \kl(of){finite path}
when creating the tape.

\begin{lemma}
  \label{lem:tape-creation}
  Let $(x_0, x_1)$ be a pair of indeterminates, $P$ be a \kl(of){finite path}
  such that $(p_0, p_1)$ is in the same orbit as $(x_0, x_1)$.
  There exists a \kl{monomial rewrite system} $R_\text{pre}$
  such that for every monomial $\monelt \in \mon{\Indets}$, the following are
  equivalent:
  \begin{enumerate}
    \item $\wenc{ \triangleright^{\text{pre}} \square \square_1 \square_2 \triangleleft^{\text{pre}}}_P
      \leftrightarrow_{R_\text{pre}}^* 
      \monelt$
      and $\wenc{\triangleright^{\text{run}}}_{P'}
      \gdivleq \monelt$ for some \kl(of){finite path} $P'$ such that
      $(p_0', p_1')$ is in the same orbit as $(x_0, x_1)$.
    \item There exists $n \geq 2$ and a \kl(of){finite path} $P'$ such that 
      $(p_0', p_1')$ is in the same orbit as $(x_0, x_1)$,
      and 
      $\monelt = \wenc{ \triangleright^{\text{run}} q_0 \square^{n}
      \triangleleft^{\text{run}} }_{P'}$.
  \end{enumerate}
  Similarly, there exists a \kl{monomial rewrite system} $R_\text{post}$
  with analogue properties using $q_f$ instead of $q_0$.
\end{lemma}
\begin{proof}
  We create the following rules,
  where $P_1$ and $P_2$ range over \kl(of){finite paths} such that
  their first two elements are in the same orbit as $(x_0, x_1)$,
  and assuming that the indeterminates of $P_1$ and $P_2$ are disjoint:
  \begin{enumerate}
    \item Cell creation: 
      $\wenc{\triangleright^{\text{pre}} \square}_{P_1}
        \wenc{ \square_1 \square_2 \triangleleft^{\text{pre}}}_{P_2}
      \leftrightarrow_{R_\text{pre}}
      \wenc{\triangleright^{\text{pre}} \square_1}_{P_1}
      \wenc{ \square \square \square_2 \triangleleft^{\text{pre}}}_{P_2}$
    \item Linearity checking:
      $\wenc{\square_1 \square}_{P_1} \wenc{\square_2 \triangleleft^{\text{pre}}}_{P_2}
      \leftrightarrow_{R_\text{pre}}
      \wenc{\square \square_1}_{P_1} \wenc{\square_2 \triangleleft^{\text{pre}}}_{P_2}$
    \item Phase transition:
      $\wenc{\triangleright^{\text{pre}} \square}_{P_1}
       \wenc{\square_1 \square_2 \triangleleft^{\text{pre}}}_{P_2}
      \leftrightarrow_{R_\text{pre}}
       \wenc{\triangleright^{\text{run}} q_0}_{P_1}
       \wenc{\square \square \triangleleft^{\text{run}}}_{P_2}$
  \end{enumerate}
  Note that there are only orbit finitely many such pairs of monomials,
  and that we can enumerate representative of these orbits because 
  $\Indets$ is \kl{effectively oligomorphic}.

  Let us first argue that this system is complete. Because there exists an
  infinite path $P_{\infty}$, it is indeed possible to reach
  $\wenc{\triangleright^{\text{run}} q_0 \square^n
  \triangleleft^{\text{run}}}_{P_\infty}$ by repeatedly applying the first
  rule, and then the second rule until $\square_1$ reaches the end of the tape,
  and continuing so until one decides to apply the third rule to reach the
  desired tape configuration.

  We now claim that the system is correct, in the sense that it can only reach
  valid tape encodings. First, let us observe that in a rewrite sequence, one
  can always assume that the rewriting takes the form of applying the first
  rule, then the second rule until one cannot apply it anymore, and repeating
  this process until one applies the third rule. Because rule (2) ensures that
  when we add new indeterminates using rule (1), they were not already present
  in the monomial, and because rule (1) ensures that locally the structure of
  the indeterminates remains a \kl(of){finite path}, we can conclude that the
  whole set of indeterminates used come from a \kl(of){finite path} $P'$. As a
  consequence, if one can reach a state where (1) or (3) are applicable, then
  the tape is of the form $\wenc{ \triangleright^{\text{pre}} \square^n
  \square_1 \square_2 \triangleleft^{\text{pre}} }_{P'}$, with $n \geq 1$. It
  follows that when one can apply rule (3), the monomial obtained is of the
  form $\wenc{ \triangleright^{\text{run}} q_0 \square^n
  \triangleleft^{\text{run}} }_{P'}$, where $P'$ is a \kl(of){finite path} such
  that $(p_0', p_1')$ is in the same orbit as $(x_0, x_1)$. 
\end{proof}

\csname thm:undecidable-paths\endcsname*
\begin{proof}
  It suffices to combine the rewriting systems $R_M$, $R_\text{pre}$ and 
  $R_\text{post}$ by taking their union.
\end{proof}

\begin{remark}
  \label{rem:multiple-orbits}
  The undecidability result of \cref{thm:undecidable-paths} can be generalised
  to a \emph{relaxed} notion of \kl(of){infinite path}. Given finitely
  many orbits $\mathcal{O}_1, \ldots, \mathcal{O}_k$ of pairs of
  indeterminates, a \emph{relaxed path} is a set of indeterminates 
  such that $(x_i, x_j)$ is belongs to one of the orbits $\mathcal{O}_k$
  if and only if $|i - j| = 1$ for all $i,j \in \N$.
\end{remark}

\begin{remark}
  \label{rem:indeterminates-infinite-path}
  Given an \kl{oligomorphic} set of indeterminates $\Indets$,
  it is equivalent to say that $\Indets$ contains an \kl(of){infinite path}
  or to say that it contains \kl(of){finite paths} of arbitrary length.
\end{remark}
\begin{proof}
  Assume that there are arbitrarily long finite paths in $\Indets$.
  Then, one can create an infinite tree whose nodes 
  are representatives of (distinct) orbits of finite paths, whose root is the empty path, and 
  where the ancestor relation is obtained by projecting on a subset of
  indeterminates.
  Because $\Indets$ is \kl{oligomorphic}, there are finitely many 
  nodes at each depth in the tree (i.e. at each length of the finite path).
  Hence, there exists an infinite branch in the tree due to König's lemma,
  and this branch is a witness for the existence of an \kl(of){infinite path}
  in $\Indets$.
\end{proof}

\begin{example}
  \label{ex:rado-graph-path}
  The \kl{Rado graph}, as introduced in \cref{ex:rado},
  contains an \kl(of){infinite path} $P$. Indeed, the Rado graph
  contains every finite graph as an induced subgraph, and in particular,
  it contains arbitrarily long finite paths. 
  As a consequence of \cref{thm:undecidable-paths},
  which applies thanks to \cref{rem:indeterminates-infinite-path},
  we conclude that the \kl{equivariant ideal membership problem}
  is undecidable for the Rado graph.
\end{example}

\begin{example}
  \label{ex:product-indets}
  Let $\Indets$ be an \kl{oligomorphic} infinite set of indeterminates.
  Then $\Indets \times \Indets$ contains a (generalised) \kl(of){infinite path}
  as defined in \cref{rem:multiple-orbits}.
\end{example}
\begin{proof}
  Let $\seqof{x_i}[i \in \N]$
  and $\seqof{y_i}[i \in \N]$ be two infinite sets of distinct indeterminates
  in $\Indets$.
  Let us define $P \defined (x_0, y_0), (x_1, y_0), (x_1, y_1), (x_2, y_1), \ldots$.
  The orbits of pairs that define the successor relation 
  are the orbits of $((x_i, y_j), (x_k, y_l))$,
  where $x_i = x_k$ and $y_j \neq y_l$, or where $x_i \neq x_k$ and $y_j = y_l$.
  Because $\Indets$ is \kl{oligomorphic}, there are finitely many such orbits.
  Let us sketch the fact that this defines a generalised path.
  Consider that
  $((x_i, y_j), (x_k, y_l))$ is in the same orbit as $((x_0, y_0), (x_1, y_0))$,
  then there exists $\gelem \in \group$ such that
  $\gelem \cdot (x_i, y_j) = (x_0, y_0)$ and $\gelem \cdot (x_k, y_l) = (x_1, y_0)$,
  but then $\gelem \cdot y_j = \gelem \cdot y_l = y_0$, and because $\gelem$ is 
  invertible, $y_j = y_l$. Similarly, we conclude that $x_i \neq x_k$.
  The same reasoning shows that if
  $((x_i, y_j), (x_k, y_l))$ is in the same orbit as $((x_0, y_0), (x_0, y_1))$,
  then $y_j \neq y_l$ and $x_i = x_k$.
\end{proof}

\section{Concluding Remarks}
\label{sec:conclusion}

We have given a sufficient condition for \kl{equivariant Gröbner bases} to be
computable, under natural \kl{computability assumptions}, and we have shown
that our sufficient condition is close to being optimal since the
undecidability of the \kl{equivariant ideal membership problem} can be derived
for a large class of group actions that do not satisfy our condition.
Let us now discuss some open questions and conjectures that arise from our work.

\paragraph{Total orderings on the set of indeterminates} We assumed that the
indeterminates $\Indets$ were equipped with a total ordering $\varleq$ that is
preserved by the group action. This assumption seems necessary, as the notions
of \kl{leading monomials} would cease to be well-defined without it. However,
we do not have a clear understanding of whether this assumption is vacuous or
not. Indeed, as noticed by \cite[Lemma 13]{GHOLAS24}, and
\cref{lem:reducts-equiv-hilbert}, it often suffices to extend the structures of
the indeterminates to account for a total ordering. A conjecture of Pouzet
\cite[Problems 12]{POUZ24} states that such an ordering always exists, and this
was remarked by \cite[Remark 14]{GHOLAS24}. Note that in this case, one would
get a complete characterisation of the group actions for which the
\kl{equivariant Hilbert basis property} holds \cite[Property 4]{GHOLAS24}.

\paragraph{Labelled well-quasi-orderings and dichotomy conjectures} As noted in
\Cref{sec:undecidability}, there are many conjectures relating the fact that
$(\mon[Y]{\Indets}, \gdivleq)$ is a \kl{well-quasi-ordering} (for every
\kl{well-quasi-ordered} set $Y$) and the presence of long
paths of some king (\cref{conj:wqo-infinite-path,rem:conj-wqo-infinite-path}).
In particular, Pouzet's conjecture \cite{POUZ72} would imply that for actions
arising from \kl{homogeneous structures} (as in the examples given in
\cref{sec:act ex}), \cref{thm:compute-egb} and \cref{thm:undecidable-paths} are
two sides of a dichotomy theorem: either the \kl{equivariant ideal membership
problem} is undecidable and there are \kl{equivariant ideals} that are not
\kl{orbit-finitely generated}, or every \kl{equivariant ideal} is
\kl{orbit-finitely generated} and one can compute \kl{equivariant Gröbner
bases}. Let us note that for some classes of graphs having bounded clique
width, Pouzet's conjecture is known to hold \cite{DRT10,LOPEZ24}.
This leads us to the following conjecture:

\begin{conjecture}
  For every action $\group\actson\Indets$ of a group $\group$ on a set
  of indeterminates that is \kl{effectively oligomorphic}, exactly
  one of the following holds:
  \begin{enumerate}
    \item The \kl{equivariant ideal membership problem} is decidable. 
    \item There exists an \kl{equivariant ideal} that is not
      \kl{orbit-finitely generated}.
  \end{enumerate}
\end{conjecture}

Let us point out that a similar conjecture was already stated in the context of
Petri nets with data. Indeed, the condition that $(\mon[Y]{\Indets},\gdivleq)$
is a \kl{WQO} for every \kl{WQO} $Y$ also guarantees coverability of Petri nets
with data $\X$ is decidable \cite[Theorem 1]{Lasota16}, and it was actually
conjectured to be a necessary condition \cite[Conjecture 1]{Lasota16}. 

\paragraph{Complexity} In the present paper, we have focused on the
decidability of the \kl{equivariant ideal membership problem} and the
computability of \kl{equivariant Gröbner bases}. However, we have not addressed
the complexity of such problems, and have only adapted the most basic
algorithms for computing \kl{Gröbner bases}. It would be interesting to know,
on the theoretical side, if one can obtain complexity lower bounds for such
problems, but also on the more practical side if advanced algorithms like
Faugère's algorithm \cite{FAUGERE02} can be adapted to the equivariant setting
and yield better performance in practice.

\paragraph*{Duals of orbit-finitely generated vector spaces} An interesting
question posed by S\l{a}womir Lasota is that whether for group actions
$\group\actson\Indets$ for which $(\mon[Y]{\Indets},\gdivleq)$ is a \kl{WQO}
for every \kl{WQO} $Y$, the duals of orbit-finitely generated vector spaces are
also orbit-finitely generated. This is a well-known fact that finitely
generated vector spaces have finitely generated duals, and the extension to the
orbit-finite setting would find applications in solving orbit-finite linear
systems of equations and extend the results of \cite{GHL22}. It is already
known to hold for actions originating from the \kl{homogeneous structures}
considered in \Cref{ex:eq atoms,ex:dlo}, via the work of
\cite{BFKM24,GHL22,Prz23}.

\bibliographystyle{plainurl}
\appendix

\end{document}